\documentclass{CSML}

\pdfoutput=1

\usepackage{lastpage}

\def\dOi{13(3:31)2017}
\lmcsheading%
{\dOi}
{1--\pageref{LastPage}}
{}
{}
{Jun.~\phantom05, 2016}
{Sep.~26, 2017}
{}

\usepackage[english]{babel}
\usepackage[utf8]{inputenc}
\usepackage[T1]{fontenc}
\usepackage{amsmath,amssymb}
\usepackage[all]{xy}
\usepackage[hidelinks]{hyperref}
\usepackage{macros}

\theoremstyle{definition}

\theoremstyle{plain}

\theoremstyle{remark}

\begin{document}
\title{Coherent Presentations of Monoidal categories}

\author[P.L.~Curien]{Pierre-Louis Curien\rsuper a}
\address{{\lsuper a}Université Paris Diderot}
\email{\href{mailto:curien@pps.univ-paris-diderot.fr}{\texttt{curien@pps.univ-paris-diderot.fr}}}
\author[S.~Mimram]{Samuel Mimram\rsuper b}
\address{{\lsuper b}École Polytechnique}
\email{\href{mailto:samuel.mimram@lix.polytechnique.fr}{\texttt{samuel.mimram@lix.polytechnique.fr}}}


\begin{abstract}
  Presentations of categories are a well-known algebraic tool to provide
  descriptions of categories by  means of generators, for objects and
  morphisms, and relations on morphisms. We generalize here this notion, in
  order to consider situations where the objects are considered modulo an
  equivalence relation,
  which is described
  by equational generators. When those form a convergent (abstract) rewriting
  system on objects, there are three very natural constructions that can be used
  to define the category which is described by the presentation: one consists in
  turning equational generators into identities (\ie considering a quotient
  category), one consists in formally adding inverses to equational generators
  (\ie localizing the category), and one consists in restricting to objects
  which are normal forms. We show that, under suitable coherence conditions on
  the presentation, the three constructions coincide, thus generalizing
  celebrated results on presentations of groups, and we extend those conditions to
  presentations of monoidal categories.
\end{abstract}

\maketitle

\tableofcontents

\section{Introduction}
Motivated by the generalization of rewriting techniques to the setting of
higher-dimensional categories, we introduce  a notion of presentation of a
monoidal category modulo a rewriting system, in order to be able to present a
monoidal category as generated by objects and morphisms, quotiented by relations
on both morphisms \emph{and objects}. This work can somehow be seen as an
extension of traditional techniques of rewriting modulo a
theory~\cite{baader1999term}: the quotient on objects is described by a
rewriting system, whose rules are called here \emph{equational}, and we want to
consider objects up to those rules. In order to handle this situation, there are
mainly two possible approaches: either implicit (work on the equivalence classes
modulo equational rules) or explicit (consider equational rules as invertible
operations).
We provide conditions on both the original rewriting system and the equational
one, so that that the two approaches coincide. Namely, we show that they imply
some form of \emph{coherence} for the equational rewriting system, \ie that
there is essentially one way of transforming an object into another using the
equational rules, thus implying that the quotient and the localization are
equivalent. An important methodological point has to be stressed  here: our
aim is not to provide the most general conditions for this to hold, but
sufficient conditions, which are applicable to a wide range of examples and can
efficiently be checked on the presentation of a (monoidal or 2-) category.

Let us further expose our motivations, which come from higher-dimensional
rewriting theory~\cite{t3rt}. A \emph{string rewriting system}~$P$ consists in
an alphabet~$P_1$ and a set~$P_2\subseteq P_1^*\times P_1^*$ of rules. Such a
system induces a monoid~$\pcat{P}={P_1^*}/{\overset*\ToT}$ obtained by
quotienting the free monoid~$P_1^*$ on~$P_1$ by the smallest congruence
$\overset*\ToT$ containing the rules in~$P_2$; when the rewriting system is
convergent, \ie both confluent and terminating, normal forms provide canonical
representatives of equivalence classes. Given a monoid~$M$, we say that $P$ is a
\emph{presentation} of~$M$ when~$M$ is isomorphic to~$\pcat{P}$: in this case,
the elements of~$P_1$ can be seen as \emph{generators} for~$M$, and the elements
of~$P_2$ as a complete set of \emph{relations} for~$M$. For instance, the
additive monoid $\N\times\N$ admits the presentation~$P$
with~$P_1=\set{\gen a,\gen b}$ and $P_2=\set{\gen{ba}\To\gen{ab}}$: namely, the
string rewriting system is convergent, and its normal forms are words of the
form $\gen{a}^p\gen{b}^q$, with $(p,q)\in\N\times\N$, from which it is easy to
build the required isomorphism.

The notion of presentation is easy to generalize from monoids to categories (a
monoid being the particular case of a category with one object): a presentation
of category consists in generators for objects and morphisms, together with
rules relating morphisms in the free category generated by the
generators. Starting from this observation, the notion of presentation was
generalized in order to present $n$-categories
(computads~\cite{street1976limits, power1991n} or
polygraphs~\cite{burroni1993higher}), thus providing us with a notion of
\emph{higher-dimensional rewriting system}. However for dimensions $n\geq 2$, this
notion of presentation has important limitations. In particular, not
every $n$-category admits a presentation. We shall illustrate this on a simple
example of a monoidal category (which is the particular case of a 2-category with
only one 0-cell).

Consider the \emph{simplicial category}~$\Delta$ whose objects are natural
numbers $p\in\N$ and morphisms $f:p\to q$ are monotone functions $f:[p]\to[q]$
where $[p]$ is the set $\set{0,\ldots,p-1}$ considered as a finite poset with
$0<\ldots<p-1$. This category is monoidal, with tensor product being given by
addition on objects ($p\otimes q=p+q$) and by ``juxtaposition'' on morphisms,
and it is well known that it admits the following presentation as a monoidal
category~\cite{mac1998categories, lafont2003towards}: its objects are generated
by one object~$\gen a$, its morphisms are generated by
$\gen m:\gen a\otimes\gen a\to\gen a$ and $\gen e:0\to\gen a$, and the relations
are
\begin{align*}
  \gen\alpha:\gen m\circ(\gen m\otimes\id_{\gen a})&=\gen m\circ(\id_{\gen a}\otimes\gen m)
  &
  \gen\lambda:\gen m\circ(\gen e\otimes\id_{\gen a})&=\id_{\gen a}
  &
  \gen\rho:\gen m\circ(\id_{\gen a}\otimes\gen e)&=\id_{\gen a}
\end{align*}
This means that every morphism of $\Delta$ can be obtained as a composite of
$\gen e$ and $\gen m$, and that two such formal composites represent the same
morphism precisely when they can be related by the congruence generated by the
above relations. As we can see on this example, a presentation~$P$ of a monoidal
category consists in generators for objects (here $P_1=\set{\gen a}$),
generators for morphisms ($P_2=\set{\gen e,\gen m}$) together with their source
and target, and relations between composites of morphisms
($P_3=\set{\gen\alpha,\gen\lambda,\gen\rho}$) together with their source and
target. Notice that such a presentation \emph{does not allow for relations
  between objects}, and thus is restricted to presenting monoidal categories
whose underlying monoid of objects is free.

This limitation can be better understood by trying to present the monoidal
category~$\Delta\times\Delta$ with tensor product extending componentwise the
one of~$\Delta$: the underlying monoid of objects is $\N\times\N$, which is not
free. If we try to construct a presentation for this monoidal category, 
we are led to
consider a presentation containing ``two copies'' of the previous presentation:
we consider a presentation~$P$ with $P_1=\set{\gen a,\gen b}$ as object
generators (where~$\gen a$ and~$\gen b$ respectively correspond to the objects
$(1,0)$ and $(0,1)$), with
$P_2=\set{\gen m_{\gen a},\gen e_{\gen a},\gen m_{\gen b},\gen e_{\gen b}}$ as
morphism generators with
\[
\gen{m_a}:\gen a\otimes\gen a\to\gen a
\qquad\qquad
\gen{e_a}:0\to\gen a
\qquad\qquad
\gen{m_b}:\gen b\otimes\gen b\to\gen b
\qquad\qquad
\gen{e_b}:0\to\gen b
\]
and with
$P_3=\set{\gen{\alpha_a},\gen{\lambda_a},\gen{\rho_a},\gen{\alpha_b},\gen{\lambda_b},\gen{\rho_b}}$
as relations. If we stop here adding relations, the presented category has
$\set{\gen a,\gen b}^*$ as underlying monoid of objects, \ie the free product of
$\N$ with itself, which is not right: recalling the above presentation for
$\N\times\N$, we should moreover add a relation~$\gen g:\gen{ba}=\gen{ab}$.
However, such a relation between objects is not allowed in the usual notion of
presentation (where only relations between morphisms are  considered). In order
to provide a meaning to it, three constructions are available:
\begin{itemize}
\item restrict~$P$ to some canonical representatives of objects modulo the
  equivalence generated by~$\gen g$ (typically the words of the form
  $\gen a^p\gen b^q$ with $(p,q)\in\N\times\N$),
\item quotient by~$\gen g$ the monoidal category~$\pcat{P}$ presented by~$P$, or
\item formally invert the morphism~$\gen g$ in~$\pcat{P}$.
\end{itemize}
We show that under reasonable assumptions on the presentation, all three
constructions coincide, thus providing a notion of \emph{coherent presentation
  modulo}. In the article, we begin by studying the case of presentations modulo
of categories and then generalize it to monoidal categories.

This article is based on the conference article~\cite{clerc2015presenting},
extending it on two major points. First, the assumptions on the opposite
presentations turned out to be unnecessary (see the new proof of
Theorem~\ref{thm:quotient-loc}), making our conditions more natural, simpler to
check and applicable to a wider range of presentations. Second, the extension to
the case of presentations of monoidal categories is  new.



%
We begin by recalling the notion of presentation of a category
(Section~\ref{sec:pres-cat}), then we extend it to work modulo a relation on
objects (Section~\ref{sec:pres-mod}), and consider the quotient and localization
\wrt to the relation (Section~\ref{sec:quotient-localization}). In order to
compare those constructions, we consider equational rewriting systems equipped
with a notion of residuation (Section~\ref{sec:residuation}) and satisfying a
particular ``cylinder'' property (Section~\ref{sec:cylinder}). We then show
that, under suitable coherence conditions, the category of normal forms is
isomorphic to the quotient (Section~\ref{sec:cat-nf}) and equivalent with the
localization (Section~\ref{sec:equiv-loc}). The notion of presentation modulo is
then generalized to monoidal categories (Section~\ref{sec:mon-pres}), as well as
the residuation techniques (Section~\ref{sec:mon-res}) and cylinder properties
(Section~\ref{sec:mon-cyl}), which finally allows us to generalize our coherence
theorem to monoidal categories (Section~\ref{sec:mon-coh}).

\bigskip
\noindent
This work was partially supported by \href{http://cathre.math.cnrs.fr/}{CATHRE
  French ANR project \hbox{ANR-13-BS02-0005-02}.} We would like to thank
Florence Clerc for her contributions to the preliminary version of this
work~\cite{clerc2015presenting}, as well as the anonymous referees for their
insightful remarks leading to improvements of this article.

\section{Presentations of categories modulo a rewriting system}
\subsection{Presentations of categories}
\label{sec:pres-cat}
Recall that a \emph{graph} $(P_0,s_0,t_0,P_1)$ consists of two sets $P_0$ and
$P_1$, of \emph{vertices} and \emph{edges} respectively, together with two
functions $s_0,t_0:P_1\to P_0$ associating to an edge its \emph{source} and
\emph{target} respectively:
\[
\vxym{
  P_0&\ar@<-.5ex>[l]_-{s_0}\ar@<.5ex>[l]^-{t_0}P_1
}
\]
Such a graph generates a category with $P_0$ as objects and the set~$P_1^*$ of
(directed) paths as morphisms. If we denote by $i_1:P_1\to P_1^*$ the coercion
of edges to paths of length~1, and $s_0^*,t_0^*:P_1^*\to P_0$ the functions
associating to a path its source and target respectively, we thus obtain a
diagram as on the left below:
\begin{equation}
  \label{eq:1-presentation}
  \vxym{
    &\ar@<-.5ex>[dl]_<<<<<{s_0}\ar@<.5ex>[dl]^<<<<<{t_0}P_1\ar[d]^{i_1}\\
    P_0&\ar@<-.5ex>[l]_{s_0^*}\ar@<.5ex>[l]^{t_0^*}P_1^*
  }
  \qquad\qquad\qquad\qquad\qquad
  \vxym{
    &\ar@<-.5ex>[dl]_<<<<<{s_0}\ar@<.5ex>[dl]^<<<<<{t_0}P_1\ar[d]^{i_1}&\ar@<-.5ex>[dl]_<<<<<{s_1}\ar@<.5ex>[dl]^<<<<<{t_1}P_2\\
    P_0&\ar@<-.5ex>[l]_{s_0^*}\ar@<.5ex>[l]^{t_0^*}P_1^*
  }
\end{equation}
in~$\Set$ which is commuting, in the sense that $s_0^*\circ i_1=s_0$ and
$t_0^*\circ i_1=t_0$.

\begin{defi}
  \label{def:presentation}
  A \emph{presentation}
  \[
  P\qeq(P_0,s_0,t_0,P_1,s_1,t_1,P_2)
  \]
  as pictured on the right of~\eqref{eq:1-presentation}, consists in a graph
  $(P_0,s_0,t_0,P_1)$ as above, the elements of $P_0$ (\resp $P_1$) being called
  \emph{object} (\resp \emph{morphism}) \emph{generators}, together with a set
  $P_2$ of \emph{relations} (or \emph{2-generators}) and two functions
  $s_1,t_1:P_2\to P_1^*$ satisfying the globular identities
  \[
  s_0^*\circ s_1=s_0^*\circ t_1
  \qquad\qquad\qquad
  t_0^*\circ s_1=t_0^*\circ t_1
  \]
  The category $\pcat{P}$ \emph{presented} by~$P$ is the category obtained from
  the category generated by the graph $(P_0,s_0,t_0,P_1)$ by quotienting
  morphisms by the smallest congruence \wrt composition identifying any two
  morphisms $f$ and $g$ such that there exists $\alpha\in P_2$ satisfying
  $s_1(\alpha)=f$ and $t_1(\alpha)=g$.
\end{defi}

\noindent
In the following, we often simply write $(P_0,P_1,P_2)$ for a presentation as
above, leaving the source and target maps implicit. We write $f:x\to y$ for an
edge $f\in P_1$ with $s_0(f)=x$ and $t_0(f)=y$, and $\alpha:f\To g$ for a
relation with $f$ as source and $g$ as target; the globular identities impose
that $f$ and $g$ have the same source (\resp target). We sometimes write
$\alpha:f\ToT g$ to indicate that $\alpha:f\To g$ or $\alpha:g\To f$ is an
element of~$P_2$, and we denote by $\overset *\ToT$ the smallest congruence such
that $f\overset*\ToT g$ whenever there exists $\alpha:f\To g$ in~$P_2$.

\begin{exa}
  The monoid $\N/2\N$ (seen as a category with only one object) admits the
  presentation~$P$ with
  \[
  P_0=\set{\gen x}
  \qquad\qquad
  P_1=\set{\gen f:\gen x\to\gen x}
  \qquad\qquad
  P_2=\set{\gen\varepsilon:\gen f\circ\gen f\To\id_{\gen x}}
  \]
\end{exa}

\noindent
Instead of considering $\overset*\ToT$ simply as a relation, it is often useful
to consider ``witnesses'' for this relation. From a categorical perspective,
this can be formalized as follows.
A presentation~$P$ generates a 2-category with invertible 2-cells (also called a
(2,1)-category), whose underlying category is the free category generated by the
underlying graph of~$P$, and whose set of 2-cells is generated by~$P_2$ and
denoted~$P_2^*$. The category presented by $P$ can be obtained from this
2-category by identifying 1-cells where there is a 2-cell in
between~\cite{burroni1993higher, lafont2003towards}. We write
$\alpha:f\overset*\ToT g$ for such a 2-cell, which provides an explicit witness
of the fact that $f$ and $g$ are identified in the presented category.

It is easily seen that any category admits a presentation:

\begin{lem}
  \label{lemma:standard-presentation}
  Any category~$\C$ admits a presentation~$P^\C$, called its \emph{standard
    presentation}, with~$P_0^\C$ being the set of objects of~$\C$, $P_1^\C$
  being the set of morphisms of~$\C$ and $P_2^\C$ being the set of pairs
  $(f_2\circ f_1,g)\in P_1^{\C*}\times P_1^{\C*}$ with $f_1,f_2,g\in P_1$ such
  that $s_0(f_1)=s_0(g)$, $t_0(f_2)=t_0(g)$ and $f_2\circ f_1=g$ in~$\C$ (with
  projections as source and target functions).
\end{lem}

\noindent
In general, a category actually admits many presentations. It can be shown that
two finite presentations present the same category if and only if they are
related by a sequence of Tietze transformations: those transformations generate
all the operations one can do on a presentation without modifying the presented
category~\cite{tietze1908topologischen,guiraud2014polygraphs}. For instance,
Knuth-Bendix completions are a particular case of
those~\cite{guiraud2013homotopical}.

\begin{defi}
  \label{def:tietze}
  Given a presentation~$P$, a \emph{Tietze transformation} consists in
  \begin{itemize}
  \item adding (\resp removing) a generator~$f\in P_1$ and a 2-generator
    $\alpha:f\To g\in P_2$ with $g\in(P_1\setminus\set{f})^*$,
  \item adding (\resp removing) a 2-generator $\alpha:f\To g\in P_2$ such that
    $f$ and $g$ are equivalent \wrt the congruence generated by the relations in
    $P_2\setminus\set{\alpha}$.
  \end{itemize}
\end{defi}



\subsection{Presentations modulo}
\label{sec:pres-mod}

In a presentation~$P$ of a category, the elements of~$P_2$ generate relations,
and the presented category is obtained by quotienting the morphisms of the free
category on the underlying graph by all these relations. We now extend this
notion in order to also allow the quotienting of objects in the process of
constructing the presented category.

\begin{defi}
  A \emph{presentation modulo} $(P,\tilde P_1)$ consists of a presentation
  $P=(P_0,P_1,P_2)$ together with a set $\tilde P_1\subseteq P_1$, whose
  elements are called \emph{equational generators}.
\end{defi}

%

\noindent
The morphisms of $\fcat{P}$ generated by the equational generators are called
\emph{equational morphisms}. Intuitively, the category presented by a
presentation modulo should be the ``quotient category'' $\pcat{P}/\tilde P_1$,
as explained in the next section, where objects equivalent under~$\tilde P_1$
(\ie related by equational morphisms) are identified.
We believe that the reason why presentations modulo of categories were not
introduced before is that they are actually unnecessary, in the sense that we
can always convert a presentation modulo into a regular presentation, see
Lemma~\ref{lemma:presentation-without-modulo} below. However, the techniques
developed here extend in the case of monoidal categories where it is not the
case anymore, see Section~\ref{sec:monoidal}, and moreover our framework already
enables one to obtain interesting results on presented categories (such as the
equivalence between quotient and localization, see~\cite{clerc2015presenting}
for details).  In this article, we will focus more on the case of presentations of monoidal
categories.

\begin{defi}
  \label{def:pres-quotient}
  \label{def:pres-demodulo}
  Given a presentation modulo $(P,\tilde P_1)$, we define the \emph{quotient
    presentation}~$P/\tilde P_1$ as the (non-modulo) presentation
  $(P'_0, P'_1, P'_2)$ where
  \begin{itemize}
  \item $P'_0 = P_0 /{\cong_1}$ where $\cong_1$ is the smallest equivalence
    relation on~$P_0$ such that $x\cong_1 y$ whenever there exists a generator
    $f:x\to y$ in $\tilde P_1$, and we denote by $\qclass x$ the equivalence
    class of $x \in P_0$,
  \item the elements of $P'_1$ are $f:\qclass x \to \qclass y$ for $f : x \to y$
    in~$P_1$,
  \item the elements of $P'_2$ are of the form $\alpha:f\To g$ for
    $\alpha:f \To g$ in $P_2$, or $\alpha_f:f\To\id_{\qclass x}$ for $f:x\To y$
    in $\tilde P_1$.
  \end{itemize}
\end{defi}

\noindent
We will sometimes consider presentations modulo with ``arrows reversed'':

\begin{defi}
  Given a presentation modulo $(P,\tilde P_1)$, the \emph{opposite
    presentation modulo} $(P^\op,\tilde P_1^\op)$ is given by
  $P^\op = (P_0, P_1^\op, P_2^\op)$, where
  $P_1^\op = \setof{ f^\op : y \to x}{f : x \to y \in P_1}$ and where
  $P_2^\op = \setof{ \alpha ^\op : f^\op \To g ^\op}{\alpha : f \To
    g}$ with $f^\op = f_1^\op \circ ... \circ f_k^\op$ for
  $f = f_k \circ\ldots\circ f_1$, and where $\tilde P_1^\op$ is the
  subset of $P_1^\op$ corresponding to $\tilde P_1$.
\end{defi}

\subsection{Quotient and localization of a presentation modulo}
\label{sec:quotient-localization}

As explained above, we want to quotient our presentations modulo by equational
morphisms, in order for the equational morphisms to induce equalities in the
presented category. Given a category~$\C$ and a set $\Sigma$ of morphisms, there
are essentially two canonical ways to ``get rid'' of the morphisms of $\Sigma$
in~$\C$: we can either force them to be identities, or to be isomorphisms,
giving rise to the following two notions of quotient and localization of a
category. These are standard constructions in category theory and we recall them
below.

\begin{defi}
  \label{def:quotient}
  The \emph{quotient} of a category~$\C$ by a set $\Sigma$ of morphisms of~$\C$
  is a category~$\C/\Sigma$ together with a \emph{quotient functor}
  $Q:\C\to\C/\Sigma$ sending the elements of $\Sigma$ to identities, such that
  for every functor $F:\C\to\D$ sending the elements of $\Sigma$ to identities,
  there exists a unique functor~$\tilde F$ such that $\tilde F\circ Q=F$.
  \[
  \vxym{
    \C\ar[d]_Q\ar[r]^F&\D\\
    \C/\Sigma\ar@{.>}[ur]_{\tilde F}&
  }
  \]
\end{defi}

\noindent
Such a quotient category always exists for general
reasons~\cite{bednarczyk1999generalized} and is unique up to isomorphism. Given
a presentation modulo $(P,\tilde P_1)$, the category presented by the associated
(non-modulo) presentation $P/\tilde P_1$ described in
Definition~\ref{def:pres-demodulo}, corresponds to considering the category
presented by the (non-modulo) presentation~$P$ and quotient it by $\tilde P_1$.
\begin{lem}
  \label{lemma:presentation-without-modulo}
  \label{lem:presentation-without-modulo}
  For every presentation modulo $(P,\tilde P_1)$, the categories
  $\pcat{P}/\tilde{P_1}$ and $\pcat{P/\tilde P_1}$ are isomorphic.
\end{lem}
\begin{proof}
  It is enough to show that $\pcat{P/\tilde P_1}$ is a quotient of $\pcat{P}$ by
  $\tilde P_1$. We define a quotient functor $Q:\pcat{P}\to\pcat{P/\tilde P_1}$
  on generators by $Q(x)=\qclass x$ for $x\in P_0$ and $Q(f)=f$ for~$f\in P_1$:
  this extends to a functor since for every 2-generator $\alpha\in P_2$ there is
  a corresponding 2\nbd{}generator in $P/\tilde P_1$. For every generator
  $f\in\tilde P_1$, we immediately have $Q(f)=\id$.
  Suppose given a functor $F:\pcat{P}\to\C$ sending equational morphisms to
  identities. We define a functor $\tilde F:\pcat{P/\tilde P_1}\to\C$ sending an
  object $\qclass x$ of $\pcat{P/\tilde P_1}$ to $\tilde F\qclass x=F x$. This
  does not depend on the choice of the representative of the class:
  given two representatives $y,y'\in\qclass x$, there exists a zig-zag of
  equational morphisms from~$y$ to~$y'$, all of which are sent by $F$ to
  identities, \ie $Fy=Fy'$.
  Given a morphism $f=f_k\circ\ldots\circ f_1$ in $\pcat{P/\tilde P_1}$ with
  $f_i\in P_1$, we define $\tilde Ff=Ff_k\circ\ldots\circ Ff_1$. For similar
  reasons, this is also well-defined.
  %
  %
  The functor~$\tilde F$ satisfies $F=\tilde F\circ Q$, and it is the only such
  functor:
  given an object $[x]$ of $\pcat{P/\tilde P_1}$, one has necessarily
  $\tilde F[x]=\tilde F\circ Q(x)=F x$ and similarly, given a generating
  morphism $f$ in $\pcat{P/\tilde P_1}$, one has necessarily
  $\tilde Ff=\tilde F\circ Q(f)=Ff$.
\end{proof}


A second, slightly different construction, consists in turning elements
of~$\Sigma$ into isomorphisms (instead of identities):

\begin{defi}
  \label{def:localization}
  The \emph{localization} of a category~$\C$ by a set $\Sigma$ of morphisms is a
  category~$\loc\C\Sigma$ together with a \emph{localization functor}
  $L:\C\to\loc\C\Sigma$ sending the elements of~$\Sigma$ to isomorphisms, such
  that for every functor $F:\C\to\D$ sending the elements of $\Sigma$ to
  isomorphisms, there exists a unique functor~$\tilde F$ such that
  $\tilde F\circ L=F$.
  \[
  \vxym{
    \C\ar[d]_L\ar[r]^F&\D\\
    \loc\C\Sigma\ar@{.>}[ur]_{\tilde F}&
  }
  \]
\end{defi}



\begin{rem}
  Note that there is a canonical functor
  \[
    \tilde Q
    \qcolon
    \loc\C\Sigma
    \qto
    \C/\Sigma
  \]
  between the localization and the quotient, induced by the universal property
  of the localization applied to the quotient functor.
\end{rem}

\noindent
In the case where the category is presented, its localization admits the
following presentation.

\begin{lem}
  \label{lem:presentation_loc_inverse}
  \label{lem:pres-loc}
  Given a presentation $P=(P_0,P_1,P_2)$ and a subset $\Sigma$ of $P_1$, the
  category presented by $P'=(P_0,P'_1,P'_2)$ where
  \[
  P'_1\qeq P_1 \uplus \setof{ \ol f : y \to x}{f : x \to y \in \Sigma}
  \]
  and
  \[
  P'_2\qeq P_2\uplus\setof{\ol f\circ f\To\id, f\circ\ol f\To\id}{f\in\Sigma}
  \]
  is a localization of the category $\pcat P$ by $\Sigma$.
\end{lem}
\begin{proof}
  The localization functor~$L$ is defined by $Lx=x$ for~$x\in P_0$, and
  $L f = f$ for $f\in P_1^*$. This functor is well-defined since for any
  2-generator $\alpha:f \To g$ in $P_2$, we have that $L f = f$ and $L g = g$,
  and there is a 2-generator $f \To g$ in $P'_2$ by definition. Besides, for any
  $f$ in $\Sigma$, $L f = f$ is an isomorphism since $\ol f$ is an inverse
  for~$f$.
  %
  Suppose given $F:\pcat{P}\to\C$ sending the elements of $\Sigma$ to
  isomorphisms. We define a functor $\tilde F:\pcat{P'} \to \C$ on the
  generators by $\tilde F x = F x$ for $x\in P_0$, $\tilde F f = F f$ for
  $f\in P_1$ and $\tilde F \ol f = (F f)^{-1}$. This functor is well-defined,
  since for any 2-generator $\alpha : f \To g$ in $P_2 \subset P'_2$, we have
  $\tilde F f = F f = F g = \tilde F g$ and
  $\tilde F (f \circ \ol f) = F f \circ F \ol f = F f \circ (F f)^{-1} = \id$
  and similarly $\tilde F (\ol f \circ f) = \id$. This functor satisfies
  $\tilde F\circ L=F$ and is the unique such functor.
\end{proof}

\begin{exa}
  \label{ex:loc-vs-quot}
  Let us consider the category
  \[
  \C\qeq\vxym{x\ar@<.5ex>[r]^-f\ar@<-.5ex>[r]_-g&y}
  \]
  with two objects and two non-trivial morphisms. Its localization by
  $\Sigma=\set{f,g}$ is equivalent to the category with one object and $\Z$ as
  set of morphisms (with addition as composition), whereas its quotient
  by~$\Sigma$ is the category with one object and only the identity as
  morphism. Notice that they are not equivalent.
\end{exa}

\noindent
The description of the localization of a category provided by the universal
property is often difficult to work with. When the set~$\Sigma$ has nice
properties, the localization admits a much more tractable
description~\cite{gabriel1967calculus,borceux1994handbook}.

\begin{defi}
  \label{def:left-calculus-of-fractions}
  A set $\Sigma$ of morphisms of a category~$\C$ is a \emph{left calculus of
    fractions} when
  \begin{enumerate}
  \item the set $\Sigma$ is closed under composition : for $f$ and $g$
    composable morphisms in $\Sigma$, $g \circ f$ is in $\Sigma$.
  \item $\Sigma$ contains the identities $\id _x$ for $x$ in $P_0$.
  \item for every pair of coinitial morphisms $u:x\to y$ in $\Sigma$ and
    $f:x\to z$ in $\C$, there exists a pair of cofinal morphisms $v:z\to t$
    in~$\Sigma$ and $g:y\to t$ in~$\C$ such that $v\circ f=g\circ u$.
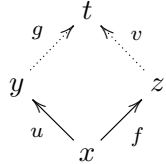
    \[
    \svxym{
      &t&\\
      y\ar@{.>}[ur]^g&&\ar@{.>}[ul]_vz\\
      &\ar[ul]^ux\ar[ur]_f&
    }
    \]
  \item for every morphism $u:x\to y$ in $\Sigma$ and pair of parallel
    morphisms $f,g:y\to z$ such that $f\circ u=g\circ u$ there exists a morphism
    $v:z\to t$ in $\Sigma$ such that $v\circ f=v\circ g$.
  \end{enumerate}
  \[
  \vxym{
   x\ar[r]^u&y\ar@<.5ex>[r]^f\ar@<-.5ex>[r]_g&z\ar@{.>}[r]^v&t
  }
  \]
\end{defi}

\begin{rem}
  \label{rem:epi-cf}
  Note that the last condition is always satisfied when every
  morphism~$u\in\Sigma$ is epi, since in this case we can take~$v$ to be the
  identity.
\end{rem}

\begin{thm}[\cite{gabriel1967calculus,borceux1994handbook}]
  \label{thm:category-of-fractions}
  When $\Sigma$ is a left calculus of fractions for a category~$\C$, the
  localization~$\loc\C\Sigma$ can be described as the \emph{category of
    fractions}, whose objects are the objects of~$\C$ and morphisms from $x$ to
  $y$ are equivalence classes of pairs of cofinal morphisms $(f,u)$ with
  $f:x\to i\in\C$ and $u:y\to i\in\Sigma$ under the equivalence relation
  identifying two such pairs $(f_1,u_1)$ and $(f_2,u_2)$ when there exists two
  morphisms $w_1,w_2\in\Sigma$ such that $w_1\circ u_1=w_2\circ u_2$ and
  $w_1\circ f_1=w_2\circ f_2$, as shown on the left below:
  \begin{equation}
    \label{eq:fractions}
    \svxym{
      &i_1\ar[d]^(.6){w_1}&\\
      x\ar[ur]^{f_1}\ar[dr]_{f_2}&j&\ar[dl]^{u_2}\ar[ul]_{u_1}y\\
      &i_2\ar[u]_(.6){w_2}&\\
    }
    \qquad\qquad\qquad\qquad
    \svxym{
      &&k&&\\
      &i\ar@{.>}[ur]^h&&\ar@{.>}[ul]_wj\\
      x\ar[ur]^f&&\ar[ul]_uy\ar[ur]^g&&\ar[ul]_vz
    }
  \end{equation}
  The identity on an object $x$ is the equivalence class of $(\id_x,\id_x)$ and
  the composition of two morphisms $(f,u):x\to y$ and $(g,v):y\to z$ is the
  equivalence class of $(h\circ f,w\circ v):x\to z$ where the morphisms $h$ and
  $w$ are provided by property 1 of
  Definition~\ref{def:left-calculus-of-fractions}, as shown on the right above.
  The localization functor $L:\C\to\loc\C\Sigma$ is the identity on objects and
  sends an morphism $f:x\to y$ to $(f,\id_y)$.
\end{thm}

\begin{exa}
  This construction draws its name from the following example. Consider the
  category~$\mathcal{Z}$ with one object $\ast$, whose morphisms are
  integers~$n\in\Z$, composition is given by multiplication and identity
  is~$1$. The set $\Sigma=\Z\setminus\set{0}$ of all morphisms excepting~$0$ is
  a left calculus of fractions since (1) it is closed under multiplication,
  (2)~it contains~$1$, (3) every pair of non-zero integers admits a non-zero
  common multiple, and it satisfies (4) by Remark~\ref{rem:epi-cf} since every
  element can be canceled. The associated localization $\loc{\mathcal{Z}}\Sigma$
  is the category with one object~$\ast$, morphisms being rational numbers
  in~$\mathbb{Q}$, with multiplication as composition: a morphism $(f,u)$ in the
  localization corresponds to the fraction $f/u$ and the quotient to the usual
  one, identifying $(fw)/(uw)$ with $f/u$.
\end{exa}

Given a presentation modulo, when the (abstract) rewriting system on objects
given by the equational generators is convergent, normal forms for objects
provide canonical representatives of objects modulo equational generators, and
therefore we are actually provided with three possible and equally reasonable
constructions for the category presented by a presentation
modulo~$(P,\tilde P_1)$:
\begin{enumerate}
\item the full subcategory~$\nfcat{P}{\tilde P_1}$ of $\pcat{P}$ whose objects
  are normal forms \wrt $\tilde P_1$,
\item the quotient category $\pcat{P}/\tilde P_1$,
\item the localization $\loc{\pcat{P}}{\tilde P_1}$.
\end{enumerate}
The aim of the following two sections is to provide reasonable assumptions on
the presentation modulo ensuring that the  first two categories are isomorphic
(normal forms provide a concrete description of the quotient), and equivalent to
the third one (which captures the coherence of equational morphisms). We
introduce these assumptions gradually in the next section. We first give some
examples illustrating the fact that those constructions are not the same in
general.

\begin{exa}
  Consider the category~$\C$ with two objects and two non-identity
  morphisms depicted on the left:
  \[
  \C=
  \vxym{
    x\ar@<.5ex>[r]^f\ar@<-.5ex>[r]_g&y
  }
  \qquad\quad
  \mathcal{T}=
  \vxym{
    y\ar@(ur,dr)^\id
  }
  \qquad\quad
  \mathcal{N}=
  \vxym{
    y\ar@(ur,dr)^{0,1,2,\ldots}
  }
  \qquad\quad
  \mathcal{Z}=
  \vxym{
    y\ar@(ur,dr)^{\ldots,-1,0,1,\ldots}
  }
  \]
  It admits a presentation~$P$ with $P_0=\set{x,y}$,
  $P_1=\set{f:x\to y,g:x\to y}$ and $P_2=\emptyset$.
  We also write~$\mathcal{T}$ for the terminal category (with one object and one
  identity morphism) and~$\mathcal{N}$ (\resp $\mathcal{Z}$) for the category
  with one object and the additive monoid $\N$ (\resp $\Z$) as monoid of
  endomorphisms.
  Taking $\tilde P_1=\set{f,g}$, we have (see also Example~\ref{ex:loc-vs-quot})
  \[
  \nfcat{P}{\tilde P_1}
  =
  \pcat{P}/\tilde P_1
  =
  \mathcal{T}
  \qquad\qquad
  \loc\C\Sigma\cong\mathcal{Z}
  \]
  Taking $\tilde P_1=\set{f}$, we have
  \[
  \nfcat{P}{\tilde P_1}=\mathcal{T}
  \qquad\qquad
  \loc\C\Sigma\cong\pcat{P}/\tilde P_1=\mathcal{N}
  \]
Thus  the three constructions are not equivalent in general. Note
  that in both cases, $\tilde P_1$ (or its closure under composition and
  identities) is not a left calculus of fractions, because condition (3) of
  Definition~\ref{def:left-calculus-of-fractions} is not satisfied.
\end{exa}

\begin{exa}
  \label{ex:huet}
  Consider the category admitting a presentation~$P$ with
  \[
  P_0=\set{x,x',y,y'}
  \qquad\qquad
  P_1=\set{f:x\to x',g:x\to y,g':y\to x,h:y\to y'}
  \]
  \ie graphically
  \begin{equation}
    \label{eq:ex-huet-graph}
    \vxym{
      x'&\ar[l]_-fx\ar@/^/[r]^-g&\ar@/^/[l]^-{g'}y\ar[r]^-h&y'
    }
  \end{equation}
  and relations
  \[
  P_2=\set{f\circ g'\circ g\To f,h\circ g\circ g'\To h}
  \]
  This presentation is a direct translation in
  our setting of the classical example in abstract rewriting systems showing
  that local confluence does not necessarily imply confluence~\cite{huet1980confluent}.
  Consider the set~$\tilde P_1=\set{g,g'}$ of equational morphisms. The quotient
  category is
  \[
  \pcat{P}/\tilde P_1
  \qeq
  \xymatrix{
    x'&\ar[l]_-fx\ar[r]^-h&y'
  }
  \]
  The localization admits the presentation given by Lemma~\ref{lem:pres-loc},
  with morphism generators
  \[
  f:x\to x'
  \qquad
  g:x\to y
  \qquad
  g':y\to x
  \qquad
  h:y\to y'
  \qquad
  \ol g:y\to x
  \qquad
  \ol g':x\to y
  \]
  and relations
  \[
  f\circ g'\circ g\To f
  \qquad
  h\circ g\circ g'\To h
  \qquad
  \ol g\circ g\To\id_x
  \qquad
  g\circ\ol g\To\id_y
  \qquad
  \ol g'\circ g'\To\id_y
  \qquad
  g'\circ\ol g'\To\id_x
  \]
  By Knuth-Bendix completion, these relations can be completed with the
  following derivable relations:
  \[
  f\circ\ol g\To f\circ g'
  \qquad
  h\circ\ol g'\To h\circ g
  \]
  (for instance the first relation can be derived by
  $f\circ\ol g=f\circ g'\circ g\circ\ol g=f\circ g'$), giving rise to a
  convergent rewriting system. The localization has the normal form
  $g'\circ g:x\to x$ as non-trivial endomorphism on~$x$, whereas all
  endomorphisms of the quotient are trivial: hence here too, quotienting is not
  equivalent to localizing.
\end{exa}

\section{Confluence properties}
\label{sec:confl}
In this section, we introduce local conditions that can be seen as a
generalization of classical local confluence properties in our context, in which
rewriting rules correspond to equational generators only, and in which we keep
track of 2-cells witnessing local confluence.

\subsection{Residuation}
\label{sec:residuation}
We begin by extending to our setting the notion of residual, which is often
associated to a confluent rewriting system in order to ``keep track'' of
rewriting steps once others have been performed~\cite{levy1978reductions,
  terese, dehornoy2013foundations}.

\begin{asm}
  \label{apt:convergence}
  We suppose fixed a presentation modulo $(P,\tilde P_1)$ such that
  \begin{enumerate}
  \item for every pair of distinct coinitial generators $f:x\to y_1$ in
    $\tilde P_1$ and $g:x\to y_2$ in $P_1$, there exist a pair of cofinal
    morphisms $g':y_1\to z$ in $P_1^*$ and $f':y_2\to z$ in $\tilde P_1^*$ and a
    2-generator $\alpha:g'\circ f\ToT f'\circ g$ in~$P_2$:
    \begin{equation}
      \label{eq:residuation-tile}
      \vxym{
        y_1\ar@{.>}[r]^{g'}&z\\
        x\ar[u]^-f\ar[r]_-g\ar@{}[ur]|-{\displaystyle\overset\alpha\Longleftrightarrow}&y_2\ar@{.>}[u]_{f'}
      }
    \end{equation}
  \item there is no infinite path with generators in~$\tilde P_1$.
  \end{enumerate}
\end{asm}

\noindent
These assumptions ensure in particular that the (abstract) rewriting system on
vertices with~$\tilde P_1$ as set of rules is convergent. Given a
vertex~$x\in P_0$, we write $\nf{x}$ for the associated normal form, \ie the
unique object $\nf{x}$ such that there is a morphism $f:x\to\nf{x}$ in
$\tilde P_1^*$ and there is no generator of the form $f:\nf{x}\to x'$
in~$\tilde P_1$. The classical Newman's lemma~\cite{newman1942theories} holds in
our framework:

\begin{lem}
  \label{lem:newman}
  For any pair of coinitial morphisms $f:x\to y_1$ in $\tilde P_1^*$ and
  $g:x\to y_2$ in $P_1^*$, there exist a pair of cofinal morphisms $g':y_1\to z$
  in $P_1^*$ and $f':y_2\to z$ in $\tilde P_1^*$ and a 2-cell
  $\alpha:g'\circ f\overset*\ToT f'\circ g$ in~$P_2^*$.
\end{lem}

\noindent
For every pair of distinct morphisms $(f,g)$ as in the Assumption 1, we
suppose fixed an arbitrary choice of a particular triple $(g',\alpha,f')$
associated to it, and write $g/f$ for $g'$, $f/g$ for $f'$ and $\rho_{f,g}$ for
$\alpha$:
\[
\vxym{
  y_1\ar@{.>}[r]^{g/f}&z\\
  x\ar[u]^-f\ar[r]_-g\ar@{}[ur]|-{\displaystyle\overset{\rho_{f,g}}\Longleftrightarrow}&y_2\ar@{.>}[u]_{f/g}
}
\]
The morphism $g/f$ (\resp $f/g$) is called the \emph{residual} of $g$ after $f$
(\resp $f$ after $g$): intuitively, $g/f$ corresponds to what remains of~$g$
once~$f$ has been performed. It is natural to extend this definition to paths as
follows:

\begin{defi}
  \label{def:residual}
  Given two coinitial paths $f:x\to y$ and $g:x\to z$ and $P_1^*$ such that
  either~$f$ or $g$ is in $\tilde P_1^*$, we define the \emph{residual} $g/f$
  of~$g$ after $f$ as above when $f$ and $g$ are distinct generators, and by
  means of the following rules:  
  \[\begin{array}{l}
f/f=\id_y\qquad\qquad
  g/\id_x=g
  \qquad\qquad
  \id_x/f=\id_y
  \\
  (g_2\circ g_1)/f=(g_2/(f/g_1))\circ(g_1/f)
  \qquad
  g/(f_2\circ f_1)=(g/f_1)/f_2
\end{array}
  \]
  (by convention the residual $g/f$ is not defined when neither~$f$ nor~$g$
  belongs to $\tilde P_1^*$). Graphically,
  \[
  \vxym{
    \ar[r]^\id&\\
    \ar[u]^g\ar[r]_{\id}&\ar[u]_{g/\id=g}
  } \qquad
  \vxym{
    \ar[r]^f & \\
    \ar[u]^{\id} \ar[r]_f & \ar[u]_{\id /f = \id}
  } \qquad
  \vxym{
    \ar[r]^{f / (g_2 \circ g_1)} & \\
    \ar[u]^{g_2} \ar[r]^{f/g_1} & \ar[u]_{g_2 / (f/g_1)} \\
    \ar[u]^{g_1} \ar[r]_f & \ar[u]_{g_1 /f}
  } \qquad
  \vxym{
    \ar[r]^{f_1 /g} & \ar[r]^{f_2 /(g /f_1)} & \\
    \ar[u]^g \ar[r]_{f_1} & \ar[r]_{f_2} \ar[u]^{g / f_1} & \ar[u]_{(g/f_1)/f_2}
  }
  \]
\end{defi}

\noindent
The above rules, when applied from left to right, provide a non-deterministic
algorithm for computing residuals of paths along paths.  We will show in
Lemma~\ref{lem:residual-well-defined} that, under an additional assumption, this
algorithm teminates and that the result does not depend on the order in which
the above rules are applied. Moreover, it can be checked that residuation is
compatible with associativity and identity laws, so that altogether the notion
of residuation is well-defined on paths.
%

\begin{rem}
  In condition (1) of Assumption~\ref{apt:convergence}, in order for Newman's
  lemma (and in fact also all subsequent properties) to hold, it would be enough
  to suppose that we have $g'\circ f\overset*\ToT f'\circ g$ instead of
  requiring that there is exactly one 2-generator~$\alpha$ mediating the two
  morphisms. However, this would makes some formulations  more involved, without
  bringing more generality in practice.
\end{rem}

\begin{rem}
  It might seem at first that Assumption~\ref{apt:convergence} is sufficient to
  ensure that quotienting by~$\tilde P_1$ or localizing \wrt $\tilde P_1$ give
  rise to equivalent categories, but Example~\ref{ex:huet} shows that this is
  not the case and more assumptions are needed. In particular termination, which
  is introduced below.
\end{rem}

To ensure that the definition of residuation is well-founded, and thus always defined, we will
make the following additional assumption. We first recall that a poset
$(N,\leq)$ is \emph{noetherian} if there is no infinite descending chain
$n_0>n_1>n_2>\ldots$ of elements of $N$; the typical example of such a poset
is~$(\N,\leq)$. A \emph{noetherian monoid} $(N,+,0,\leq)$ is a (non-necessarily
commutative) monoid $(N,+,0)$ together with a structure of noetherian poset
$(N,\leq)$, such that for every $x,y,y',z\in N$,
\[
y>y'
\qquad\text{implies}\qquad
x+y+z>x+y'+z
\]
and $0$ is the minimum element. Again, a typical example of such a monoid is
$(\N,+,0,\leq)$.

\begin{asm}
  \label{apt:1-termination}
  There is a weight function $\omega_1:P_1\to N$, where $(N,+,0,\leq)$ is a
  noetherian monoid, such that for every generator $g\in P_1$ and
  $f\in\tilde P_1$, we have $\omega_1(g/f)<\omega_1(g)$, where we extend
  $\omega_1$ on elements of $P_1^*$ by
  $\omega_1(g\circ f)=\omega_1(g)+\omega_1(f)$ and $\omega_1(\id)=0$.
\end{asm}

\begin{rem}
  Note in particular that, with the previous assumption, we always have
  \[
  \omega_1(g)<\omega_1(h)+\omega_1(g)+\omega_1(f)=\omega_1(h\circ g\circ f)
  \]
  for composable morphisms $f$, $g$ and $h$.
\end{rem}


In order to study confluence of the rewriting system provided by equational
morphisms, through the use of residuals, we first introduce the following
category, which allows us to consider, at the same time, both residuals~$g/f$
and~$f/g$ of two coinitial morphisms~$f$ and~$g$.

\begin{defi}
  \label{def:zig-zag}
  The \emph{zig-zag presentation} associated to the presentation modulo
  $(P,\tilde P_1)$ is the presentation $Z=(Z_0,Z_1,Z_2)$ with $Z_0=P_0$,
  $Z_1 = P_1\uplus\ol{\tilde P_1}$ (generators in $\ol{\tilde P_1}$ are of the
  form $\ol f:B\to A$ for any generator $f:A\to B$ in~$\tilde P_1$) and
  relations in $Z_2$ are of the form
  $g\circ\ol f\To \overline{(f/g)}\circ (g/f)$
  \[
  \vxym{
    y_1\ar[r]^{g/f}&z\\
    x\ar@{<-}[u]^-{\ol f}\ar[r]_-g\ar@{}[ur]|-{\Longrightarrow}&y_2\ar@{<-}[u]_{\ol{f/g}}
  }
  \]
  or $f\circ\ol f\To\id_y$ for any pair of distinct coinitial generators
  $f:x\to y\in\tilde P_1$ and $g:x\to z\in P_1$.
\end{defi}

\begin{lem}
  \label{lem:residuation-zigzag}
  The rewriting system on morphisms in $Z_1^*$ with $Z_2$ as rules is
  convergent. Given two coinitial morphisms $f:x\to y$ in $\tilde P_1^*$ and
  $g:x\to z$ in $P_1^*$, the normal form of~$g\circ\ol f$ is
  $\ol{(f/g)}\circ(g/f)$.
\end{lem}
\begin{proof}
  We extend the weight function of Assumption~\ref{apt:1-termination} to
  morphisms in~$Z_1^*$ by setting $\omega_1(\ol f)=0$ for~$\ol{f}$
  in~$\ol{\tilde P_1}$. This ensures that the rewriting system on morphisms
  in~$Z_1^*$ with~$Z_2$ as rules is terminating. Moreover, because the left
  members of rules are of the form $g\circ\ol{f}$ with $g\in P_1$ and
  $\ol{f}\in\ol{\tilde P_1}$, there are no critical pairs (a morphism of the
  form $g\circ\ol{f}$ cannot non-trivially overlap with a morphism of the form
  $g'\circ\ol{f'}$), which implies that the rewriting system is confluent.
  Given two coinitial morphisms~$f:x\to y$ in~$\tilde P_1^*$ and~$g:x\to z$
  in~$P_1^*$, we prove by well-founded induction on~$\omega_1(g\circ\ol f)$ that
  the normal form of $g\circ\ol f$ is $\ol{(f/g)}\circ(g/f)$.
  If either~$f$ or~$g$ is an identity, this  is direct.
  Otherwise, $f = f_2 \circ f_1$ and $g = g_2 \circ g_1$ where $f_1$, $f_2$,
  $g_1$ and $g_2$ are non-identity morphisms.
  \[
  \vxym{
      \ar[r]^{(g_1 / f_1) /f_2} & \ar[rr]^{g_2 / (f/g_1)} && \\
    &&& \\
    \ar[uu]^{f_2} \ar[r]^{g_1 / f_1}\ar@{}[uur]|{\overset*\To}  & \ar[uu]_{f_2 / (g_1 / f_1)}&& \\
    \ar[u]^{f_1} \ar[r]_{g_1}\ar@{}[ur]|{\overset*\To}&\ar[rr]_{g_2} \ar[u]_{f_1/g_1}\ar@{}[uuurr]|{\overset*\To} && \ar[uuu]_{(f/g_1)/g_2}
  }\]
  By induction, we have
  \[
  g_1 \circ \ol {f_1}\overset*\To \ol{(f_1 / g_1)} \circ (g_1 / f_1)
  \qquad\text{and}\qquad
  (g_1 / f_1) \circ \ol {f_2}\overset*\To \ol{(f_2 / (g_1/f_1))} \circ ((g_1 / f_1)/f_2)
  \]
  because
  \[
  \omega_1(g_1\circ\ol{f_1})
  \quad<\quad
  \omega_1(g_2\circ g_1\circ\ol{f_1}\circ\ol{f_2})
  \qeq
  \omega_1(g\circ\ol{f})
  \]
  and
  \[
  \omega_1((g_1 / f_1) \circ \ol {f_2})
  \quad<\quad
  \omega_1\pa{g_2 \circ\ol{(f_1 / g_1)} \circ (g_1 / f_1)\circ \ol f_2}
  \quad<\quad
  \omega_1 (g \circ \ol f)
  \]
  Therefore,
  \begin{align*}
    g \circ \ol f
    & \overset * \qTo g_2 \circ \ol{(f_1 / g_1)} \circ (g_1 / f_1)\circ \ol f_2 \\
    & \overset * \qTo g_2 \circ \ol{(f_1 / g_1)} \circ \ol{(f_2 / (g_1 / f_1))} \circ ((g_1 / f_1)/f_2)\\
    & \qeq g_2 \circ \ol{(f / g_1)} \circ (g_1 / f)
  \end{align*}
  Similarly,
  \[
  \omega_1 (g_2 \circ \ol{(f / g_1)}\circ\ol{(f_2/(g_1/f_1))})
  \quad<\quad
  \omega_1 (g \circ \ol f)
  \]
  therefore
  \[
  g_2 \circ \ol{(f / g_1)}\circ\ol{(f_2/(g_1/f_1))} \overset*\qTo \ol{((f/g_1)/g_2)} \circ (g_2 / (f/g_1))
  \]
  and we have
  \begin{align*}
    g \circ \ol f
    & \overset * \qTo g_2 \circ \ol{(f / g_1)} \circ (g_1 / f) \\
    & \overset * \qTo \ol{((f/g_1)/g_2)} \circ (g_2 / (f/g_1)) \circ (g_1 / f) \\
    & \qeq \ol{(f/g)} \circ (g/f)
  \end{align*}
  from which we conclude.
\end{proof}

\begin{rem}
  \label{rem:1-termination}
  The termination Assumption~\ref{apt:1-termination} is not the only possible
  one. For instance, an abstract rewriting system is called \emph{strongly
    confluent} when $x\to y_1$ and $x\to y_2$ implies that there exists $z$ such
  that $y_1\to z$ (or $y_1=z$) and $y_2\overset*\to z$. Such an abstract
  rewriting system is always confluent~\cite{huet1980confluent}. This translates
  to our setting: if, in every residuation relation of the
  form~\eqref{eq:residuation-tile}, we have that~$f/g$ (\resp $g/f$) is always a
  generator or an identity, then the rewriting system on $Z_1^*$ with $Z_2$ as
  rules is confluent and $g\circ\ol f$ rewrites to $\ol{(f/g)}\circ(g/f)$.
\end{rem}

\noindent
As a direct corollary of the convergence of the rewriting system, one can show
that Definition~\ref{def:residual} makes sense:

\begin{lem}
  \label{lem:residual-well-defined}
  The residuation operation does not depend on the order in which equalities of
  Definition~\ref{def:residual} are applied.
\end{lem}

\noindent
Moreover, a ``global'' version of the residuation property
(Assumption~\ref{apt:convergence}) holds:

\begin{prop}
  \label{prop:residual-relation}
  Given two coinitial morphisms $f:x\to y$ in $\tilde P_1^*$ and $g:x\to z$ in
  $P_1^*$, there exists a 2-cell
  $\alpha:(g/f)\circ f\overset *\ToT(f/g)\circ g$.
\end{prop}
\begin{proof}
  By Lemma~\ref{lem:residuation-zigzag}, there exists a rewriting path
  $\beta : g \circ \ol f \To \ol{(f / g)} \circ (g/f)$ in $Z_2^*$.
  By induction on its length, we can construct a 2-cell
  $\alpha : (g / f) \circ f \overset * \ToT (f/g) \circ g$ in the following
  way. The case where~$\beta$ is empty is immediate, otherwise we have
  $f=f_2\circ f_1$ and $g = g_2 \circ g_1$ where $f_2$ is in $\tilde P_1 ^*$
  (\resp $g_2$ in $P_1 ^*$) and $f_1$ is a generator in $\tilde P_1$ (\resp
  $g_1$ in $P_1$). We distinguish two cases depending on the form of the first
  rule of~$\beta$:
  \[
  \vxym{
    \ar[rr]^{g / f = g_2 / f_2} && \\
    \ar[u]^{f_2} \ar[r]^{\id} & & \\
    \ar[u]^{f_1} \ar[r]_{g_1}\ar@{}[ur]|{\To} & \ar[r]_{g_2} \ar[u]_{\id} & \ar[uu]_{f / g = f_2 / g_2}
  }
  \qquad\qquad\qquad\qquad
  \vxym{
    \ar[rrr]^{g/f} &&& \\
    &&& \\
    \ar[uu]^{f_2} \ar[r]^{g_1 / f_1} & \ar[rr]^{g_2 /(f_1/g_1)}&& \ar[uu]_{f_2 / (g/f_1)} \\
    \ar[u]^{f_1} \ar[r]_{g_1}\ar@{}[ur]|{\overset{\phantom{*}}\To} & \ar[rr]_{g_2} \ar[u]_{f_1/g_1}\ar@{}[urr]|{\overset*\To} && \ar[u]_{f_1 /g}
  }
  \]
  If $f_1 = g_1$, \ie if the first step of~$\beta$ corresponds to rewriting
  $g_2 \circ g_1 \circ \ol{f_1} \circ \ol{f_2}$ to $g_2 \circ \ol{f_2}$ by
  applying the rewriting rule $f_1\circ\ol{f_1}\To \id$ of~$Z_2$ (we necessarily
  have $f_1=g_1$), then by induction hypothesis, there exists a 2-cell
  \[
  \alpha'
  \qcolon
  (g_2 / f_2) \circ f_2
  \quad\overset*\ToT\quad
  (f_2 / g_2) \circ g_2
  \]
  Since $f_2 / g_2 = f/g$ and $g_2 / f_2 = g/f$, this means that there exists a
  2-cell
  \[
  (g/f)\circ f
  \quad\overset*\ToT\quad
  (f/g)\circ g
  \]
  Otherwise $f_1 \neq g_1$, and $g_2 \circ g_1 \circ \ol{f_1} \circ \ol{f_2}$
  rewrites to $g_2 \circ \ol{(f_1/g_1)} \circ (g_1 /f_1) \circ \ol{f_2}$ by
  applying the rewriting rule
  $g_1 \circ \ol{f_1} \To \ol{(f_1/g_1)} \circ (g_1 /f_1)$ of $Z_2$. By
  definition of the 2-generators in~$Z_2$, there exists a 2-generator
  \[
  (g_1/f_1) \circ f_1
  \quad\ToT\quad
  (f_1 / g_1) \circ g_1
  \]
  in $P_2$. Moreover, by Lemma~\ref{lem:residuation-zigzag}, the morphism
  $g_2 \circ \ol{(f_1/g_1)}$ in~$Z_1^*$ rewrites to
  $\ol{(f_1 / g)} \circ (g_2 / (f_1 /g_1))$, and therefore by induction
  hypothesis, there exists a 2-cell
  \[
  (g_2 /(f_1 / g_1)) \circ (f_1 / g_1)
  \quad\overset*\ToT\quad
  ((f_1 / g_1) /g_2) \circ g_2
  \]
  in $P_2^*$. This means that there is a 2-cell in~$P_2^*$
  \[
  (g/f_1)\circ f_1
  \qeq
  (g_2 / (f_1 / g_1)) \circ (g_1 /f_1) \circ f_1
  \quad\overset*\ToT\quad
  ((f_1 / g_1)/g_2) \circ g_2 \circ g_1 
  \qeq
  (f_1 / g) \circ g
  \]
  Similarly, by lemma \ref{lem:residuation-zigzag}, $(g / f_1) \circ \ol {f_2}$
  rewrites to $\ol{(f_2 / (g / f_1)} \circ (g /f)$ by rules in $Z_2$, which
  means that there exists a 2-cell
  \[
  (g / f) \circ f_2
  \quad\overset*\ToT\quad
  (f_2 / (g / f_1)) \circ (g / f_1)
  \]
  in~$P_2^*$ and therefore, there exists a 2-cell in~$P_2^*$:
  \[
  (g / f) \circ f = (g / f) \circ f_2 \circ f_1 \overset * \ToT (f_2 / (g / f_1)) \circ (f_1 / g) \circ g = (f / g) \circ g
  \]
  from which we conclude, as indicated in the above diagram.
\end{proof}

\subsection{The cylinder property}
\label{sec:cylinder}
In  Section \ref{sec:residuation}, we have studied residuation, which enables one to
recover a residual $g/f$ of a morphism $g$ after a coinitial equational
morphism~$f$ (and similarly for $f/g$). We now strengthen our hypothesis in
order to ensure that if two morphisms are equal (\wrt the equivalence generated
by $P_2^*$) then their residuals after a same morphism are equal, \ie equality
is compatible with residuation.

\bigskip
\noindent
\begin{asm}
  \label{apt:cylinder}
  The presentation $(P,\tilde P_1)$ satisfies the \emph{cylinder property}: for
  every triple of coinitial morphism generators $f:x\to x'$ in $\tilde P_1$
  (\resp in $P_1$) and $g_1,g_2:x\to y$ in $P_1^*$ (\resp in $\tilde P_1^*$)
  such that there exists a generating 2-cell $\alpha:g_1\ToT g_2$, we have
  $f/g_1=f/g_2$ and there exists a 2-cell $g_1/f\overset*\ToT g_2/f$. We write
  $\alpha/f$ for an arbitrary choice of such a 2-cell.
  \begin{equation}
    \label{eq:cylinder}
    \vxym{
      x'\ar@/^/@{.>}[rr]^{g_1/f}\ar@/_/@{.>}[rr]_{g_2/f}\ar@{}[rr]|-{\alpha/f}&&y'\\
      \ar[u]^fx\ar@/^/[rr]^{g_1}\ar@/_/[rr]_{g_2}\ar@{}[rr]|-\alpha&&\ar@{.>}[u]_{f/g_1=f/g_2}y
    }
  \end{equation}
\end{asm}

\noindent
As in the previous section, we would like to extend this ``local'' property ($f$
and~$\alpha$ are supposed to be generators) to a ``global'' one (where $f$
and~$\alpha$ can be composites of cells):

\begin{prop}[Global cylinder property]
  \label{prop:global-cylinder}
  Given coinitial morphisms $f:x\to x'$ in $\tilde P_1^*$ (\resp in $P_1^*$) and
  $g_1,g_2:x\to y$ in $P_1^*$ (\resp in $\tilde P_1^*$) such that there exists a
  composite 2-cell $\alpha:g_1\overset*\ToT g_2$, we have $f/g_1=f/g_2$ and
  there exists a 2-cell $g_1/f\overset*\ToT g_2/f$.
\end{prop}

\noindent
The proof of the previous proposition requires generalizing, in dimension 2, the
termination condition (Assumption~\ref{apt:1-termination}) and the construction
of the zig-zag presentation (Definition~\ref{def:zig-zag}).

\begin{defi}
  The \emph{2-zig-zag presentation} associated to $(P,\tilde P_1)$ is
  $Y=(Y_0,Y_1,Y_2)$ with
  \begin{itemize}
  \item $Y_0=P_0$,
  \item $Y_1=\horiz{P_1}\uplus\verti{P_1}$ where $\horiz{P_1}=\verti{P_1}=P_1$,
    the superscripts ``$\textnormal{H}$'' and ``$\textnormal{V}$'' being used to
    distinguish between the two copies of the disjoint union: the morphisms of
    $\horiz{P_1}$ are called \emph{horizontal}, and noted $\horiz{f}:A\to B$ for
    some morphism $f:A\to B$ in $P_1$, and similarly for the morphisms in
    $\verti{P_1}$ which are called \emph{vertical}, and
  \item the
  2-cells in $Y_2=\horiz Y_2\uplus\verti Y_2$ are either\\
  \begin{itemize}
  \item horizontal 2-cells: $\horiz Y_2=\horiz{P_2}\uplus\horiz{\ol{P_2}}$ (\ie
    2-generators in $P_2$ taken forward or backward, and decorated
    by~$\mathrm{H}$), or
  \item vertical 2-cells: given two generators $f:x\to y$ and $g:x\to z$ in
    $P_1$ such that~$f$ or~$g$ belongs to~$\tilde P_1$, we have a 2-generator
    $\verti\rho_{f,g}:\horiz{(g/f)}\circ\verti{f} \To
    \verti{(f/g)}\circ\horiz{g}$ in~$\verti Y_2$.
    \[
    \vxym{
      x'\ar[r]^{\horiz{(g/f)}}&y'\\
      x\ar[u]^{\verti f}\ar[r]_{\horiz{g}}\ar@{}[ur]|-{\displaystyle\overset{\verti\rho_{f,g}}\Longrightarrow}&\ar[u]_{\verti{(f/g)}}
    }
    \]
  \end{itemize}
  \end{itemize}  
\end{defi}

\noindent
We consider the following rewriting system on the 2-cells in~$Y_2^*$ of the
2-category generated by the presentation: for every
1-cell $f:x\to x'$ in $P_1$ and coinitial generating 2-cell
$\alpha:g_1\ToT g_2:x\to y$ in~$P_2$, such that either $f$ or both $g_1$ and
$g_2$ belong to $\tilde P_1^*$, there is a rewriting rule
\begin{equation}
  \label{eq:2-zigzag-rs}
  \begin{array}{rcl}
  (\verti{(f/g_1)}\circ\horiz\alpha)\vcirc\verti\rho_{f,g_1}
  &\qTO&
  \verti\rho_{f,g_2}\vcirc(\horiz{(\alpha/f)}\circ\verti f)
  \\
  \vxym{
    x'\ar@/^/[rr]^{\horiz{g_1/f}}&\ar@{}[d]|-<<{\verti\rho_{f,g_1}}&y'\\
    \ar[u]^{\verti{f}}x\ar@/^/[rr]^{\horiz g_1}\ar@/_/[rr]_{\horiz g_2}\ar@{}[rr]|-{\horiz\alpha}&&\ar[u]_{\verti{(f/g_1)}}y
  }
  &\qTO&
  \vxym{
    x'\ar@/^/[rr]^{\horiz{(g_1/f)}}\ar@/_/[rr]_{\horiz{(g_2/f)}}\ar@{}[rr]|-{\horiz{(\alpha/f)}}&\ar@{}[d]|->>{\verti\rho_{f,g_2}}&y'\\
    \ar[u]^{\verti{f}}x\ar@/_/[rr]_{\horiz g_2}&&\ar[u]_{\verti{(f/g_1)}}y
  }
  \end{array}
\end{equation}
where $\circ$ (\resp $\vcirc$) denotes horizontal (\resp vertical) composition
in a 2-category. 

In order to ensure the termination of the rewriting system, we suppose the
following.

\begin{asm}
  \label{apt:2-termination}
  There is a weight function $\omega_2:\horiz{P_2}\to N$, where $N$ is a
  noetherian commutative monoid, such that for every $\alpha:g_1\To g_2$
  in~$\horiz{P_2}$ and $f$ in $P_1$ such that $\alpha/f$ exists, we have
  $\omega_2(\alpha/f)<\omega_2(\alpha)$, where $\omega_2$ is extended to
  $\pa{\horiz{P_2}\uplus\horiz{\ol{P_2}}}^*$ by
  $\omega_2(\ol\alpha)=\omega_2(\alpha)$, $\omega_2(\id)=0$, and both horizontal
  and vertical compositions are sent to addition.
  %
\end{asm}

\noindent
The assumption that the ordered monoid~$N$ is commutative ensures that the
definition of~$\omega_2$ is compatible with the axioms of 2-categories, such as
associativity or exchange law.

\begin{cor}
  The rewriting system~\eqref{eq:2-zigzag-rs} is convergent.
\end{cor}

\begin{rem}
  \label{rem:2-termination}
  In a similar way as in Remark~\ref{rem:1-termination}, the
  Assumption~\ref{apt:2-termination} is not the only possible one. Depending on
  the presentation, variants can be more adapted. For instance, if the residual
  $f/g_1=f/g_2$ of the vertical morphism~$f$ in a cylinder~\eqref{eq:cylinder}
  is always a generator or an identity, then the rewriting
  system~\eqref{eq:2-zigzag-rs} is confluent, which is weaker than the previous
  corollary but sometimes sufficient in practice. Also, notice that there are
  really two kinds of cylinders~\eqref{eq:cylinder} considered here: those for
  which $f$ is equational and those for which $g_1$ and $g_2$ are both
  equational. Both cases can be handled separately, \ie two different weights
  (or methods) can be used to handle each of the two cases.
\end{rem}

\noindent
Proposition~\ref{prop:global-cylinder} follows easily, by a reasoning similar to
Proposition~\ref{prop:residual-relation}.

\medskip
The cylinder property has many interesting consequences for the residuation
operation, as we now investigate.

\begin{prop}
  \label{prop:epi}
  In the category $\pcat{P}$, every equational morphism is epi.
\end{prop}
\begin{proof}
  Suppose given $f:x\to y$ in $\tilde P_1^*$, and $g_1,g_2:y\to z$ in $P_1^*$
  such that $g_1\circ f\overset*\ToT g_2\circ f$. By
  Proposition~\ref{prop:global-cylinder}, we have
  \[
  g_1
  \qeq
  (g_1\circ f)/f
  \quad\overset*\ToT\quad
  (g_2\circ f)/f
  \qeq
  g_2
  \]
  from which we conclude.
\end{proof}

\noindent
Our axiomatization can also be used to show the following
proposition, which will not be used in the rest of
the article:

\begin{prop}[\cite{clerc2015presenting}]
  In the category $\pcat{P}$, every morphism~$g$ admits a pushout along a
  coinitial equational morphism $f$ given by $g/f$.
\end{prop}

\begin{rem}
  The careful reader will have noticed that, so far, we have only used the
  cylinder property in the case where the ``vertical morphism'' $f$ is
  equational. The case where both $g_1$ and $g_2$ are equational will be used in
  the proof of Theorem~\ref{thm:nf-quotient}.
\end{rem}

\section{Comparing presented categories}
\subsection{The category of normal forms}
\label{sec:cat-nf}
We first show that with our hypotheses on the rewriting system, the quotient
category~$\pcat{P}/\tilde P_1$ can be recovered as the following subcategory of
$\pcat{P}$, whose objects are those which are in normal form for $\tilde P_1$.

\begin{defi}
  \label{def:cat-nf}
  The \emph{category of normal forms}~$\nfcat{P}{\tilde P_1}$ is the full
  subcategory of~$\pcat{P}$ whose objects are the normal forms of elements of
  $P_0$ \wrt rules in $\tilde P_1$. We write
  $I:\nfcat{P}{\tilde P_1}\to\pcat{P}$ for the inclusion functor.
\end{defi}

\noindent
For every object $x$ of $\pcat{P}$, we shall denote the associated normal form
by $\nf{x}$, and for every such object~$x$ we shall fix a choice of an
equational morphism $u_x$ from~$x$ to its normal form. Note that, by Newman's
Lemma~\ref{lem:newman}, if $u_x':x\to\nf{x}$ is another choice of such a
morphism then there is a 2-cell $u_x\overset*\ToT u_x'$. Also, we always have
$u_{\nf{x}}=\id_{\nf{x}}$.

\begin{thm}
  \label{thm:nf-quotient}
  The category $\nfcat{P}{\tilde P_1}$ is isomorphic to the quotient category
  $\pcat P/\tilde P_1$.
\end{thm}
\begin{proof}
  We show that the category $\nfcat{P}{\tilde P_1}$ is a quotient of $P$ by
  $\tilde P_1$. We define a functor~$N:\pcat{P}\to\nfcat{P}{\tilde P_1}$ as the
  functor associating to each object~$x$ its normal form $\nf{x}$
  under~$\tilde P_1$, and to each morphism $f:x\to y$, the morphism
  $\nf{f}:\nf{x}\to\nf{y}$ where $\nf{f}=u_{y'}\circ(f/u_x)$ with $y'$ being the
  target of $f/u_x$:
  \[
  \xymatrix{
    & \hat{y} = \nf y' \\
    \hat{x} \ar@/^/@{..>}[ur]^{\nf{f}} \ar[r]^{f / u_x} & y' \ar[u]_{u_{y'}}\\
    x \ar[u]^{u_x} \ar[r]_f & y \ar[u]_{u_x / f}
  }
  \]
  Notice that, a priori, this definition depends on a choice of a representative
  in $P_1^*$ for $f$, and in $\tilde P_1^*$ for~$u_x$ and~$u_{y'}$, in the
  equivalence classes of morphisms modulo the relations in $P_2$. The global
  cylinder property shown in Proposition~\ref{prop:global-cylinder} ensures that
  the definition is independent of the choice of such representatives (in
  particular, for~$u_x$ we use the consequence of Newman's lemma mentioned above
  and the cylinder property in the case where the basis is equational).
  Given two composable morphisms $f:x\to y$ and $g:y\to z$ we have
  \[
  \begin{array}{r@{\ =\ }l}
    Ng\circ Nf
    &u_{z'}\circ(g/u_y)\circ u_{y'}\circ (f/u_x)\\
    &u_{z'}\circ(g/(u_{y'}\circ(u_x/f)))\circ u_{y'}\circ (f/u_x)\\
    &u_{z'}\circ (g/(u_x/f))/u_{y'}\circ u_{y'}\circ (f/u_x)\\
    &u_{z'}\circ u_{y'}/(g/(u_x/f))\circ g/(u_x/f)\circ (f/u_x)\\
    &u_{z''}\circ((g\circ f)/u_x)\\
    &N(g\circ f)
  \end{array}
  \vcenter
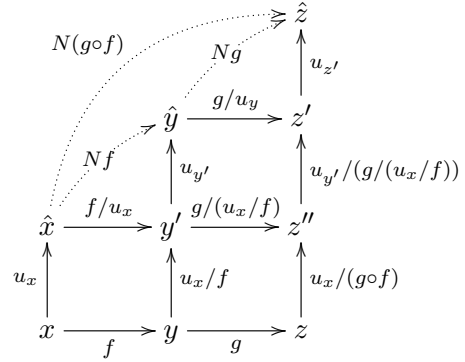
{
    \xymatrix@C=7ex{
      && \hat{z} \\
      & \hat{y} \ar@/^/@{.>}[ur]|{Ng} \ar[r] ^{g / u_y} & z' \ar[u]_{u_{z'}} \\
      \hat{x} \ar@/^6ex/@{.>}[uurr]^{N(g\circ f)} \ar@/^/@{.>}[ur]|{Nf} \ar[r]^{f / u_x} & y' \ar[u]_{u_{y'}}\ar[r]^{g/(u_x/f)} & z''\ar[u]_{u_{y'}/(g/(u_x/f))}\\
      x \ar[u]^{u_x} \ar[r]_f & y \ar[u]_{u_x / f} \ar[r]_g & z \ar[u]_{u_x/(g\circ f)}
    }
  }
  \]
  The image of an equational morphism~$u:x\to y$ under the functor $N$ is an
  identity. Namely, we have $Nu=\nf{u}=u_{y'}\circ(u/u_x)$, with
  $u/u_x:\nf{x}\to y'$: since $u/u_x$ is an equational morphism (as the residual
  of an equational morphism) whose source is a normal form, necessarily
  $u/u_x=\id_{\nf{x}}$, $y'=\nf{x}$ and $u_{y'}=\id_{\nf{x}}$. In particular,
  $N$ preserves identities.

  Suppose given a functor $F:\pcat{P}\to\C$ sending the equational morphisms to
  identities. We have to show that there exists a
  unique functor $G:\nfcat{P}{\tilde P_1}\to\C$ such that $G\circ N=F$.
  Writing $I:\nfcat{P}{\tilde P_1}\to\pcat{P}$ for the inclusion functor, it is
  easy to show $I$ is a section of~$N$, \ie
  $N\circ I=\Id_{\nfcat{P}{\tilde P_1}}$, and we define $G=F\circ I$.
  \[
  \vxym{
    \pcat{P}\ar[d]^N\ar[r]^F&\C\\
    \ar@/^/[u]^I\nfcat{P}{\tilde P_1}\ar@{.>}[ur]_G
  }
  \]
  Since $F$ sends equational morphisms to identities, it is easy to check that
  $G\circ N=F$: given an object $x$, we have
  \[
  G\circ N(x)=G(\nf{x})=F\circ I(\nf{x})=F(\nf{x})=F(x)
  \]
  the last equality, being due to the fact that
  $F(u_x)=\id_{F(\nf{x})}=\id_{F(x)}$, and similarly for morphisms.
  %
  %
  Finally, we check the uniqueness of the functor~$G$. Suppose given another
  functor $G':\nfcat{P}{\tilde P_1}\to\C$ such that $G'\circ N=F=G\circ N$. We
  have $G'=G'\circ N\circ I=G\circ N\circ I=G$.
\end{proof}

\subsection{Equivalence with the localization}
\label{sec:equiv-loc}
We now show that the two previous constructions (quotient and normal forms) also
coincide with the third possible construction which consists in formally adding
inverses for equational morphisms.

\begin{defi}
  \label{def:pres-coh}
  A presentation modulo $(P,\tilde P_1)$ is called \emph{coherent} when the
  canonical functor~$\loc{\pcat P}{\tilde P_1}\to\pcat P/\tilde P_1$ is an
  equivalence of categories.
\end{defi}

\noindent
First, notice that we can use the description of the localization
$\loc{\pcat{P}}{\tilde P_1}$ as a category of fractions given in
Theorem~\ref{thm:category-of-fractions}:

\begin{lem}
  The set $\fcat{\tilde P_1}/P_2$ of equational morphisms of~$\pcat{P}$ is a left
  calculus of fractions.
\end{lem}
\begin{proof}
  We have to show that the set of equational morphisms satisfies the four
  conditions of Definition~\ref{def:left-calculus-of-fractions}: the first two
  (closure under composition and identities) are immediate, the third one
  follows from Proposition~\ref{prop:residual-relation}, and the last one is
  ensured by the fact that all equational morphisms are epi by Proposition
  \ref{prop:epi}, see Remark~\ref{rem:epi-cf}.
\end{proof}


\begin{thm}
  \label{thm:quotient-loc}
  A presentation modulo $(P,\tilde P_2)$ which satisfies
  assumptions~\ref{apt:convergence} to \ref{apt:2-termination} is coherent.
\end{thm}
\begin{proof}
  By Theorem~\ref{thm:nf-quotient}, the statement can be rephrased as the claim
  that $\nfcat{P}{\tilde P_1}$ and $\loc{\pcat P}{\tilde P_1}$ are equivalent
  categories.
  
  Suppose given a morphism $(f,u)$ from $x$ to~$\nf{y}$ in the category of
  fractions~$\loc{\pcat P}{\tilde P_1}$, where~$\nf{y}$ is a normal form
  under~$\tilde P_1$, as on the left below
  \[
  \vxym{
    x\ar[r]^f&i&\ar[l]_u\nf{y}
  }
  \qquad\qquad\qquad\qquad
  \svxym{
    &i_1\ar[d]^(.6){w_1}&\\
    x\ar[ur]^{f_1}\ar[dr]_{f_2}&j&\ar[dl]^{u_2}\ar[ul]_{u_1}\nf{y}\\
    &i_2\ar[u]_(.6){w_2}&\\
  }
  \]
  Since $u$ is equational and $\nf{y}$ is a normal form, one necessarily has
  $i=\nf{y}$ and $u=\id_{\nf{y}}$. Similarly, given two equivalent morphisms
  $(f_1,u_1)$ and $(f_2,u_2)$ whose targets are both a normal form $\nf{y}$, as
  on the right above, one has $i_1=i_2=\nf{y}$ and
  $u_1=u_2=w_1=w_2=\id_{\nf{y}}$, and therefore~$f_1=f_2$.
  Now, consider the functor $F:\nfcat{P}{\tilde P_1}\to\loc{\pcat P}{\tilde P_1}$
  defined as the composite of the inclusion functor
  $I:\nfcat{P}{\tilde P_1}\to\pcat{P}$, see Definition~\ref{def:cat-nf}, with
  the localization functor~$L:\pcat{P}\to\loc{\pcat P}{\tilde P_1}$, see
  Definition~\ref{def:localization}:
  \[
  \vxym{
    \nfcat{P}{\tilde P_1}\ar[dr]_I\ar[rr]^F&&\loc{\pcat P}{\tilde P_1}\\
    &\pcat{P}\ar[ur]_L
  }
  \]
  The functor~$F$ sends a morphism $f:\nf{x}\to\nf{y}$ in the category of normal
  forms to the morphism~$(f,\id_{\nf{y}})$ in the category of fractions. The
  preceding remarks imply immediately that the functor~$F$ is full and faithful.
  %
  %
  Finally, given an object~$y\in\loc{\pcat P}{\tilde P_1}$, there is a
  morphism~$u:y\to\hat{y}$ in $\tilde P_1^*$ to its normal form which induces an
  isomorphism $y\cong\hat{y}$ in $\loc{\pcat P}{\tilde P_1}$. The functor~$F$
  thus provides a weak inverse to the canonical functor
  $\loc{\pcat P}{\tilde P_1}\to \pcat P/\tilde P_1$, which is therefore an
  equivalence of categories.
\end{proof}

\noindent
An illustration of this theorem is provided in~\cite{clerc2015presenting}, on
the presentation of a ``dihedral category'' (note that the assumptions on the
opposite presentation~$P^\op$ mentioned there were superfluous, as shown by the
new proof of the above theorem). Here, in Section~\ref{sec:monoidal}, we will
provide a detailed example, in the refined setting of a presentation of a
monoidal category.

\subsection{Embedding into the localization}
\label{sec:loc-embedding}
In this section, we show another direct application of our techniques. It is
sometimes useful to show that a category embeds into its localization. When the
category is equipped with a calculus of fractions, this can be shown using the
following proposition~\cite[Exercise~5.9.2]{borceux1994handbook}:

\begin{prop}
  \label{prop:factions-embedding}
  Given a left calculus of fractions~$\Sigma$ for a category~$\C$, all the
  morphisms of~$\Sigma$ are mono if and only if the inclusion functor
  $L:\C\to\loc\C\Sigma$ is faithful.
\end{prop}
\begin{proof}
  Suppose that the elements of~$\Sigma$ are monos. Given two morphisms
  $f_1,f_2:x\to y$ in~$\C$ such that $Lf_1=Lf_2$, we have a diagram as on the
  left of~\eqref{eq:fractions} with $u_1=u_2=\id_{y}$, and therefore
  $w_1=w_2$. Commutation of the left part of the diagram gives
  $w_1\circ f_1=w_2\circ f_2$ and therefore $f_1=f_2$ since $w_1=w_2$ is
  mono. The functor~$L$ is faithful.

  Conversely, suppose that~$L$ is faithful. Given morphisms~$w$, $f_1$ and~$f_2$
  such that $w\in\Sigma$ and $w\circ f_1=w\circ f_2$, one has $Lf_1=Lf_2$ and
  therefore $f_1=f_2$. The morphism $w$ is thus mono.
\end{proof}

\noindent
Showing that the elements of~$\Sigma$ are monos can however be difficult. In the
case where~$\C=\pcat{P}$ and $\Sigma=\fcat{P_1}$, for some presentation
modulo~$(P,\tilde P_1)$, it can be proved as follows.

\begin{lem}
  Suppose given a presentation modulo $(P,\tilde P_1)$ such that the opposite
  presentation modulo $(P^\op,\tilde P^\op)$ satisfies
  Assumptions~\ref{apt:convergence}, \ref{apt:1-termination}, \ref{apt:cylinder}
  and \ref{apt:2-termination}. Then the localization functor
  $\pcat{P}\to\loc{\pcat{P}}{\tilde P_1}$ is faithful.
\end{lem}
\begin{proof}
  By the dual of Proposition~\ref{prop:epi}, all equational morphisms are mono,
  and we apply Proposition~\ref{prop:factions-embedding}.
\end{proof}

\noindent
Again, an example of application is provided in~\cite{clerc2015presenting}.

\begin{rem}
  The result in the previous proposition is close to Dehornoy's theorem,
  see~\cite{dehornoy2000completeness}
  and~\cite[Section~II.4]{dehornoy2013foundations}, stating that a monoid with a
  presentation satisfying suitable conditions (our assumptions are variants of
  those) embeds into the enveloping groupoid.  Dehornoy's setting is more
  restricted, since taking the enveloping groupoid corresponds to localizing
  \wrt every morphism, while we consider localization with respect to a class of
  morphisms. However, we also need stronger conditions: in
  Assumption~\ref{apt:convergence}, we require the equational rewriting system
  to be terminating, which is never the case for presentations of monoids since
  they have only one object when seen as categories. Dehornoy's conditions also
  impose termination properties (called there \emph{Noetherianity}), but only
  ``locally''. A detailed comparison, together with conditions unifying the two
  approaches, is left for future work.
\end{rem}

\newcommand{\exch}{\overset\chi\Leftrightarrow} 

\section{Coherent presentations of monoidal categories}
\label{sec:monoidal}

\subsection{Presentations of monoidal categories}
\label{sec:mon-pres}
We now turn our attention to presentations of \emph{monoidal categories} and
describe how the previous developments can be adapted to this setting. Only
strict and small such categories will be considered in this article. We start
from premonoidal categories~\cite{power1997premonoidal}, which will be of some
use later on.
In fact, all the developments performed in this section could have been carried out in
the slightly more general setting of 2-categories. However, we feel that the
shift in dimension would have obscured the comparison with the previous sections.

\newcommand{\monunit}{I}
\begin{defi}
  A (strict) \emph{premonoidal category} $(\C,\otimes,\monunit)$ consists of a
  category~$\C$ together with
  \begin{enumerate}
  \item for every object $x\in\C$, a functor
    $x\otimes -:\C\to\C$ called \emph{left action},
  \item for every object $x\in\C$, a functor
    $-\otimes x:\C\to\C$ called \emph{right action},
  \item an object $\monunit\in\C$, called \emph{unit object},
  \end{enumerate}
  such that
  \begin{itemize}
  \item the left and right actions coincide on objects: for every objects
    $x,y\in\C$, $x\otimes y$ is the same whether the $\otimes$ operation is the
    left or the right action, thus justifying the use of the same notation,
  \item the set of objects of~$\C$ is a monoid when equipped with $\otimes$ as
    multiplication and $\monunit$ as neutral element: for every objects $x,y,z\in\C$,
    \[
    (x\otimes y)\otimes z=x\otimes(y\otimes z)
    \qquad\qquad
    \monunit\otimes x=x=x\otimes\monunit
    \]
  \item the left action is a monoid action: for every objects $x,y\in\C$ and
    morphism $f$,
    \[
    x\otimes(y\otimes f)=(x\otimes y)\otimes f
    \qquad\qquad
    \monunit\otimes f=f
    \]
  \item the right action is a monoid action: for every objects $x,y\in\C$ and
    morphism~$f$,
    \[
    (f\otimes x)\otimes y=f\otimes(x\otimes y)
    \qquad\qquad
    f\otimes \monunit=f
    \]
  \item the left and right actions are compatible: for every objects $x,y\in\C$
    and morphism~$f$,
    \[
    (x\otimes f)\otimes y=x\otimes(f\otimes y)
    \]
  \end{itemize}
  A (strict) \emph{monoidal category} is a premonoidal category as above
  satisfying the exchange law: for every morphisms $f:x\to x'$ and $g:y\to y'$,
  \[
  (x'\otimes g)\circ(f\otimes y)
  =
  (f\otimes y')\circ(x\otimes g)
  \]
  allowing us to denote by $f\otimes g$ this morphism, and for every objects
  $x,y\in\C$,
  \[
  x\otimes\id_y=\id_{x\otimes y}=\id_x\otimes y
  \]
  We sometimes omit the tensor and simply write $xy$ instead of $x\otimes y$.
\end{defi}

\begin{defi}
  \label{def:mon-funct}
  A \emph{monoidal functor} $F:\C\to\D$ between two (pre)monoidal categories is
  a functor equipped with a morphism $\eta:\monunit_\D\to F(\monunit_\C)$ and a
  natural transformation of components
  \[
  \mu_{x,y}
  \qcolon
  F(x)\otimes_\D F(y)
  \qto
  F(x\otimes_\C y)
  \]
  making the following diagrams commute for every $x,y,z\in\C$:
  \[
  \xymatrix@C=8ex{
    F(x)\otimes F(y)\otimes F(z)\ar[d]_{\mu_{x,y}\otimes F(z)}\ar[r]^-{F(x)\otimes\mu_{y,z}}&F(x)\otimes F(y\otimes z)\ar[d]^{\mu_{x,y\otimes z}}\\
    F(x\otimes_\C y)\otimes F(z)\ar[r]_-{\mu_{x\otimes y,z}}&F(x\otimes y\otimes z)\\
  }
  \]
  \[
  \xymatrix{
    \monunit\otimes F(x)\ar[d]_{\eta\otimes F(x)}\ar@{=}[r]&F(x)\ar@{=}[d]\\
    F(\monunit)\otimes F(x)\ar[r]_-{\mu_{\monunit,x}}&F(\monunit\otimes x)
  }
  \qquad\qquad\qquad
  \xymatrix{
    F(x)\otimes\monunit\ar[d]_{F(x)\otimes\monunit}\ar@{=}[r]&F(x)\ar@{=}[d]\\
    F(x)\otimes F(\monunit)\ar[r]_-{\mu_{x,\monunit}}&F(x\otimes\monunit)
  }
  \]
  A monoidal functor is \emph{strong} (\resp \emph{strict}) when $\eta$ and
  $\mu_{x,y}$ are isomorphisms (\resp identities).
\end{defi}

\noindent
Since giving a monoidal structure on a category adds a structure of monoid on
the objects, this suggests introducing the following generalization of graphs
and presentations, in order to present monoidal categories.

\begin{defi}
  A \emph{monoidal graph} $(P_0,s_0,t_0,P_1)$ consists of a diagram
  \[
  \vxym{
    P_0\ar[d]^{i_0}&\ar@<-.5ex>[dl]_{s_0}\ar@<.5ex>[dl]^{t_0}P_1\\
    P_0^*&
  }
  \]
  in~$\Set$, where $P_0^*$ is the free monoid on~$P_0$ and $i_0:P_0\to P_0^*$ is
  the canonical injection (sending an element to the corresponding word with one
  letter).
\end{defi}

\noindent
Note that a monoidal graph is simply another name for the data of a \emph{string
  rewriting system}: the set~$P_0$ is the \emph{alphabet}, with $P_0^*$ as set
of \emph{words} over it,
and~$P_1$ is the set of \emph{rewriting rules} along with their source and
target respectively indicated by the functions~$s_0$ and $t_0$. This allows us
to consider classical notions in string rewriting theory (such as critical
pairs) in this context, see~\cite{baader1999term,terese} for details about
those.

A monoidal graph freely generates a monoidal category. If we write $P_1^*$ for
its set of morphisms $i_1:P_1\to P_1^*$ for the canonical injection of
generators into morphisms, and $s_0^*,t_0^*:P_1^*\to P_0^*$ for the source and
target maps, we obtain a diagram
\[
\vxym{
  P_0\ar[d]^{i_0}&\ar@<-.5ex>[dl]_{s_0}\ar@<.5ex>[dl]^{t_0}P_1\ar[d]^{i_1}\\
  P_0^*&\ar@<-.5ex>[l]_{s_0^*}\ar@<.5ex>[l]^{t_0^*}P_1^*
}
\]
where $s_0^*\circ i_1=s_0$ and $t_0^*\circ i_1=t_0$. An explicit description of
the morphisms in~$P_1^*$ is given by the following lemma:
the morphisms of the free premonoidal category are easy to describe and those of
the free monoidal category can be obtained by explicitly quotienting by axioms
imposing that the exchange law holds.

\begin{lem}
  \label{lem:non-mon-pres}
  Suppose fixed a monoidal graph~$P$.
  \begin{enumerate}
  \item The underlying category of the free premonoidal category generated by $P$ is
    the free category generated (in the sense of Section~\ref{sec:pres-cat}) by
    the graph $Q$
    \begin{itemize}
    \item with $Q_0=P_0^*$ as vertices
    \item edges in~$Q_1$ are triples
      \[
      (x,f,z)
      \qcolon
      xyz
      \qto
      xy'z
      \]
      with $x,z\in P_0^*$ and $f:y\to y'$ in~$P_1$.
    \end{itemize}
    and is equipped with the expected premonoidal structure, whose left and
    right actions are given by
    \[
    x'\otimes(x,f,z)\otimes z'
    \qeq
    (x'x,f,zz')
    \]
    In the following, we write $x\otimes f\otimes z$, or even $xfz$, instead of
    $(x,f,z)$ for edges, and the morphisms in~$Q_1^*$ will be denoted
    by~$P_1^\otimes$.
  \item The underlying category of the free monoidal category generated by~$P$
    is the category presented (in the sense of
    Definition~\ref{def:presentation}) by ~$Q=(Q_0,Q_1,Q_2)$, where $Q_2$ is the
    set of all relations
    \begin{equation}
      \label{ex:exch-rel}
      \chi_{z_1fz_2,z_3gz_4}
      \qcolon
      (z_1x'z_2z_3gz_4)\circ(z_1fz_2z_3yz_4)
      \qTo
      (z_1fz_2z_3y'z_4)\circ(z_1xz_2z_3gz_4)
    \end{equation}
    called \emph{exchange relations}, where $z_1fz_2,z_3gz_4\in Q_1$, with
    $f:x\to x'$ and $g:y\to y'$ in $P_1$. We write $\exch$ for the equivalence
    relation generated by~$Q_2$. The morphisms of this category are denoted by
    $P_1^*$ and its monoidal structure is induced by the previous premonoidal
    structure.
  \end{enumerate}
\end{lem}

\begin{exa}
  Consider the monoidal graph with~$P_0=\set{\gen{a}}$ and $P_1=\set{\gen{m}:\gen{a}\gen{a}\to\gen{a}}$. The
  following are morphisms in the free (pre)monoidal category:
  \[
  \gen{ma}\circ\gen{maa}\circ\gen{aaam}
  \qquad\qquad\qquad
  \gen{am}\circ\gen{maa}\circ\gen{maaa}
  \]
  and can be represented using string diagrams as
  \[
  \fig{exm1}
  \qquad\qquad\qquad
  \fig{exm2}
  \]
  These are equal in the free monoidal category, but not in the free premonoidal
  category: one needs the exchange rules in order to transform one into the
  other.
\end{exa}

\noindent
With the notations of the previous lemma, a generator~$xfz:xyz\to xy'z$ in~$Q_1$
can also be called a \emph{rewriting step}: it corresponds to the rewriting rule
$f:y\to y'$ used in a context with the word $x$ on the left and $z$ on the
right. We sometimes write
\[
Q_1\qeq P_0^*P_1P_0^*
\]
for the set of rewriting steps.
From this point of view, the morphisms in~$P_1^\otimes$ are
\emph{rewriting paths} and the morphisms in~$P_1^*$ are rewriting paths up to
commutation of independent rewriting steps, \ie up to the equivalence relation
$\exch$.

\begin{defi}
  A \emph{monoidal presentation} $P=(P_0,s_0,t_0,P_1,s_1,t_1,P_2)$ consists of a
  diagram
  \[
  \vxym{
    P_0\ar[d]^{i_0}&\ar@<-.5ex>[dl]_{s_0}\ar@<.5ex>[dl]^{t_0}P_1\ar[d]^{i_1}&\ar@<-.5ex>[dl]_{s_1}\ar@<.5ex>[dl]^{t_1}P_2\\
    P_0^*&\ar@<-.5ex>[l]_{s_0^*}\ar@<.5ex>[l]^{t_0^*}P_1^*
  }
  \]
  in~$\Set$, where
  \begin{itemize}
  \item $P_0$ is a set of \emph{object generators};
  \item $P_0^*$ is the free monoid on~$P_0$ and $i_0:P_0\to P_0^*$ is the
    canonical injection;
  \item $P_1$ is a set of \emph{morphism generators}, with
    $s_0,t_0:P_1\to P_0^*$ indicating their source and target;
  \item $P_1^*$ is as in Lemma~\ref{lem:non-mon-pres},
    with corresponding source and target maps $s_0^*,t_0^*:P_1^*\to P_0^*$.
  \item $P_2$ is a set of \emph{relations} (or \emph{2-cell generators}), with
    $s_1,t_1:P_2\to P_1^*$ indicating their source and target, which should
    satisfy the globular identities $s_0^*\circ s_1=s_0^*\circ t_1$ and
    $t_0^*\circ s_1=t_0^*\circ t_1$.
  \end{itemize}
  The monoidal category $\pcat{P}$ \emph{presented} by~$P$ is the monoidal
  category with $P_0^*$ as set of objects and whose morphisms are the elements of
  $P_1^*$, quotiented by the smallest congruence (\wrt both composition and
  tensor product) identifying any two morphisms $f$ and $g$ such that there is a
  relation $\alpha:f\To g$.
\end{defi}

\noindent
We also introduce the notation $P_2^*$ (\resp $P_2^\otimes$) for the set of
2-cells in the monoidal (2,1)\nobreakdash-cate\-gory (\resp premonoidal (2,1)-category)
whose underlying monoidal (\resp premonoidal) category is freely generated by
the underlying monoidal graph of~$P$, and 2-cells are generated by~$P_2$. We do
not detail these constructions: all the reader needs to remember for the sequel
is that these 2-cells are formal (vertical) composites of 2-cells of the form
\begin{equation}
  \label{eq:2-gen-ctx}
  x\alpha z
  \qcolon
  xfz
  \qTo
  xgz
  \qcolon
  xyz
  \qto
  xy'z
\end{equation}
 for $x,z\in P_0^*$ and $\alpha:f\To g:y\to y'$ in~$P_2$, or their
inverses. The set of 2-cells of the form~\eqref{eq:2-gen-ctx} is denoted
$P_0^*P_2P_0^*$.

Note that a presented monoidal category has an underlying monoid of objects
which is free. Therefore, not every monoidal category admits a presentation, \eg
the category with $\N/2\N$ as monoid of objects and only identities as
morphisms. In this setting, the use of coherent presentation is really
necessary: there is no associated notion of ``quotient presentation'' (as in
Definition~\ref{def:pres-quotient}).

\begin{defi}
  A \emph{monoidal presentation modulo}  consists of a monoidal
  presentation together with a set $\tilde P_1\subseteq P_1$ of \emph{equational
    generators}  (notation $(P,\tilde P_1)$).
\end{defi}

\noindent
As before, we say that a morphism in~$P_1^*$ (or in $P_1^\otimes$) is
\emph{equational} when it can be obtained by composing and tensoring equational
generators and identities. We write $\tilde P_1^*\subseteq P_1^*$
(or~$\tilde P_1^\otimes\subseteq P_1^\otimes$) for the set of equational
morphisms.

We now generalize the notions of quotient and localization to monoidal
categories.

\begin{defi}
  The \emph{quotient} of a monoidal category~$\C$ by a set~$\Sigma$ of morphisms
  is a monoidal category~$\C/\Sigma$ together with a strict monoidal functor
  $\C\to\C/\Sigma$ sending elements of~$\Sigma$ to identities, which is
  universal with this property.
\end{defi}

\begin{defi}
  \label{def:mon-localization}
  The \emph{localization} of a monoidal category~$\C$ by a set~$\Sigma$ of
  morphisms is a monoidal category~$\loc\C\Sigma$ together with a strict
  monoidal functor $L:\C\to\loc\C\Sigma$ sending the elements of~$\Sigma$ to
  isomorphisms, which is universal with this property.
\end{defi}

\noindent
The localization of a presented monoidal category always admits a monoidal
presentation as in Lemma~\ref{lem:pres-loc}. Moreover, the description as a
category of fractions under suitable conditions
(Theorem~\ref{thm:category-of-fractions}) is still valid~\cite{day1973note}:

\begin{prop}
  Suppose given a left calculus of fractions~$\Sigma$ for a monoidal
  category~$\C$, which is closed under tensor product, \ie for every $f,g\in\Sigma$,
  we have $f\otimes g\in\Sigma$. The associated category of fractions is
  canonically monoidal and isomorphic to the localization in the sense of
  Definition~\ref{def:mon-localization}.
\end{prop}
\begin{proof}
  The unit object is the one of~$C$, and given two morphisms $(f,u)$ and $(g,v)$
  in the category of fractions $\loc\C\Sigma$, we define their tensor product as
  $(f,u)\otimes(f',u')=(f\otimes f',u\otimes u')$. Suppose that $(f_1,u_1)$ and
  $(f_2,u_2)$ (\resp $(f'_1,u'_1)$ and $(f'_2,u'_2)$) are two representatives of
  the same morphism, \ie that we have mediating morphisms as on the left and the
  middle below:
  \[
  \svxym{
    &i_1\ar[d]^(.6){w_1}&\\
    x\ar[ur]^{f_1}\ar[dr]_{f_2}&j&\ar[dl]^{u_2}\ar[ul]_{u_1}y\\
    &i_2\ar[u]_(.6){w_2}&\\
  }
  \qquad\qquad
  \svxym{
    &i_1'\ar[d]^(.6){w_1'}&\\
    x'\ar[ur]^{f_1'}\ar[dr]_{f_2'}&j'&\ar[dl]^{u_2'}\ar[ul]_{u_1'}y'\\
    &i_2'\ar[u]_(.6){w_2'}&\\
  }
  \qquad\qquad
  \svxym{
    &i_1\otimes i_1'\ar[d]|-(.4){w_1\otimes w_1'}&\\
    x\otimes x'\ar[ur]^{f_1\otimes f_1'}\ar[dr]_{f_2}&j\otimes j'&\ar[dl]^{u_2}\ar[ul]_{u_1\otimes u_1'}y\otimes y'\\
    &i_2\otimes i_2'\ar[u]|-(.4){w_2\otimes w_2'}&\\
  }
  \]
  The diagram on the right shows that $(f_1,u_1)\otimes(f_1',u_1')$ and
  $(f_2,u_2)\otimes(f_2',u_2')$ represent the same morphism. The fact that the
  axioms of a monoidal category are satisfied is easily deduced from the fact
  that $\C$ does satisfy those axioms and from the closure of $\Sigma$ under
  tensor product.
\end{proof}

\subsection{Residuation in monoidal presentations}
\label{sec:mon-res}
We now explain how to extend the residuation techniques developed in
Section~\ref{sec:confl} to presentations of monoidal categories. By
Lemma~\ref{lem:non-mon-pres}, a presentation of a monoidal category can be seen
as a presentation of a premonoidal category together with explicit exchange
rules~$\chi_{f,g}$. The general strategy is thus to apply the previous
constructions and to show that they are compatible with the exchange law: this
strategy turns out to work in our running example, but we explain in
Section~\ref{sec:mon-cyl'} that further generalizations of the axioms are
sometimes needed, requiring to deal explicitly with exchange relations.
From now on, we thus consider that~$P_2$ contains relations of the
form~\eqref{ex:exch-rel}. This of course makes the presentation infinite;
however, these relations will be handled in a specific way, and we will only
need to consider a finite number of those (by only considering ``critical
situations'').

\begin{rem}
  In fact, it is easily shown that we can restrict to relations of the
  form~\eqref{ex:exch-rel} with $z_1$, $z_3$ and $z_4$ empty (all the others can be
  deduced). We will do so in the sequel in order to simplify computations.
\end{rem}

%
As an illustrative example, we will study a presentation of a category simpler
than the example of~$\Delta\times\Delta$ mentioned in the introduction, in order
to have a smaller number of conditions to check. We consider the category
$\Delta_s$ whose objects are natural numbers $p\in\N$ and morphisms $f:p\to q$ are
surjective functions $f:\intset{p}\to\intset{q}$, with
$\intset{p}=\set{0,\ldots,p-1}$. As in the case of $\Delta$, this category is
monoidal with tensor product  given on objects by addition, and with $0$ as
neutral element (such a category is often called a PRO). As a simple variation
on the example of~$\Delta$, this category admits the following presentation.

\begin{lem}
  The category~$\Deltas$ admits the monoidal presentation~$P$ with
  \begin{align*}
    P_0&=\set{\gen a}&
    P_1&=\set{\gen m:\gen{aa}\to\gen a}&
    P_2&=\set{\gen\alpha:\gen m\circ(\gen{ma})\To\gen m\circ(\gen{am})}
  \end{align*}
\end{lem}


\begin{exa}
  \label{ex:pres-ds2}
  We are  interested in presenting the category~$\Deltas\times\Deltas$ using
  a presentation modulo. For reasons explained in the introduction, it is
  natural to expect that this category admits the monoidal presentation
  modulo~$P$ with generators
  \begin{align*}
    P_0&=\set{\gen a,\gen b}&
    P_1&=\set{\gen m:\gen{aa}\to\gen a,\gen n:\gen{bb}\to\gen b,\gen g:\gen{ba}\to\gen{ab}}
  \end{align*}
  and relations in $P_2$ being
  \begin{align*}
    \gen\alpha\qcolon \gen m\circ(\gen{ma})&\qTo\gen m\circ(\gen{am})\\
    \gen\beta\qcolon\gen n\circ(\gen{nb})&\qTo\gen n\circ(\gen{bn})\\
    \gen\gamma\qcolon\gen g\circ\gen{bm}&\qTo\gen{mb}\circ\gen{ag}\circ\gen{ga}\\
    \gen\delta\qcolon\gen g\circ\gen{na}&\qTo\gen{an}\circ\gen{gb}\circ\gen{bg}
  \end{align*}
  (plus the mandatory exchange relations) which can be depicted, in categorical
  notation, as
  \[
  \svxym{
    &\ar[dl]_{\gen{ma}}\gen{aaa}\ar[dr]^{\gen{am}}&\\
    \gen{aa}\ar[dr]_{\gen m}&\overset{\gen\alpha}\To&\ar[dl]^{\gen m}\gen{aa}\\
    &\gen a&
  }
  \qquad
  \svxym{
    &\ar[dl]_{\gen{nb}}\gen{bbb}\ar[dr]^{\gen{bn}}&\\
    \gen{bb}\ar[dr]_{\gen n}&\overset{\gen\beta}\To&\ar[dl]^{\gen n}{\gen bb}\\
    &\gen b&
  }
  \qquad
  \svxym{
    &\ar[dl]_{\gen{bm}}\gen{baa}\ar[dr]^{\gen{ga}}&\\
    \ar[ddr]_{\gen g}\gen{ba}&\ar@{}[d]|{\displaystyle\overset{\gen\gamma}\To}&\gen{aba}\ar[d]^{\gen{ag}}\\
    &&\gen{aab}\ar[dl]^{\gen{mb}}\\
    &\gen{ab}
  }
  \qquad
  \svxym{
    &\ar[dl]_{\gen{na}}\gen{bba}\ar[dr]^{\gen{bg}}&\\
    \ar[ddr]_{\gen g}\gen{ba}&\ar@{}[d]|{\displaystyle\overset{\gen\delta}\To}&\gen{bab}\ar[d]^{\gen{gb}}\\
    &&\gen{abb}\ar[dl]^{\gen{an}}\\
    &\gen{ab}
  }
  \]
  In string-diagrammatic notation, the generators $m$, $n$ and $g$ can be
  respectively drawn as
  \[
  \fig{gen_m}
  \qquad\qquad\qquad
  \fig{gen_n}
  \qquad\qquad\qquad
  \fig{gen_g}
  \]
  (the notation is the same for the  first two, but the typing of wires makes the
  notation unambiguous), and the relations can be depicted as
  \begin{align*}
    \fig{assoc_a_l}&\overset{\gen\alpha}\To\fig{assoc_a_r}&
    \fig{assoc_b_l}&\overset{\gen\beta}\To\fig{assoc_b_r}\\
    \fig{gamma_l}&\overset{\gen\gamma}\To\fig{gamma_r}&
    \fig{delta_l}&\overset{\gen\delta}\To\fig{delta_r}
  \end{align*}
  The set of equational generators is $\tilde P_1=\set{\gen g}$.
\end{exa}

First, consider Assumption~\ref{apt:convergence} on our presentation modulo. The
first condition of this assumption asserts that coinitial rewriting steps are
confluent whenever one of them is equational. However, we now have a monoidal
structure and the exchange axioms provide obvious ways to close diagrams in many
cases. For instance, given an equational generator $f:x\to x'$ and a generator
$g:y\to y'$, we can always show the confluence of the pair $(fy,xg)$ of
coinitial morphisms:
\begin{equation}
  \label{eq:res-exchange}
  \vxym{
    x'y\ar@{.>}[r]^{x'g}\ar@{}[dr]|{\displaystyle\overset{\chi_{f,g}}\Longleftrightarrow}&x'y'\\
    xy\ar[u]^{fy}\ar[r]_{xg}&xy'\ar@{.>}[u]_{fy'}
  }
\end{equation}
Moreover, whenever there is a diagram as on the left, there is also one as on
the right:
\begin{equation}
  \label{eq:res-context}
  \vxym{
    y_1\ar@{.>}[r]^{g'}\ar@{}[dr]|{\displaystyle\overset\alpha\Longleftrightarrow}&y'\\
    y\ar[u]^f\ar[r]_g&\ar@{.>}[u]_{f'}y_2
  }
  \qquad\qquad\qquad
  \vxym{
    xy_1z\ar@{.>}[r]^{xg'z}\ar@{}[dr]|{\displaystyle\overset{x\alpha z}\Longleftrightarrow}&xy'z\\
    xyz\ar[u]^{xfz}\ar[r]_{xgz}&\ar@{.>}[u]_{xf'z}xy_2z
  }
\end{equation}
For this reason, one only has to ensure that diagrams can be closed for pairs of
coinitial morphisms which are ``minimal'' (\wrt left and right context) and not
in exchange position.
This observation is well-known in rewriting theory, and used to show that, in a
string rewriting system, the confluence of critical pairs implies local
confluence, which is reformulated in Lemma~\ref{lem:mon-local-confl}
below. This suggests adapting Assumption~\ref{apt:convergence} as follows.

\begin{defi}
  A pair of coinitial rewriting steps $f:x\to y$ and $g:x\to z$ is called a
  \emph{critical pair} when
  \begin{itemize}
  \item $f$ and $g$ are distinct,
  \item for every pair of coinitial rewriting steps $f':x'\to y'$ and
    $g':x'\to z'$ and words $u$ and $v$ such that $f=uf'v$ and
    $g=ug'v$, the words $u$ and $v$ are empty,
  \item there is no pair of rewriting steps $f':x'\to y'$ and $g':x''\to z'$
    such that $f=f'x''$ and $g=x'g'$,
  \item the previous condition also holds if we exchange the roles of~$f$
    and~$g$.
  \end{itemize}
\end{defi}

\setcounter{assumption}{0}
\begin{asm}
  We suppose fixed a presentation modulo $(P,\tilde P_1)$ such that
  \begin{enumerate}
  \item for every pair of coinitial rewriting steps $f:x\to y_1$
    in~$P_0^*\tilde P_1P_0^*$ and $g:x\to y_2$ in~$P_0^*P_1P_0^*$ forming a
    critical pair, there exists a pair of cofinal morphisms $g/f:y_1\to z$
    in~$P_1^\otimes$ and $f/g:y_2\to z$ in $\tilde P_1^\otimes$ and a generating
    2-cell $\alpha:g/f\circ f\ToT f/g\circ g$ in~$P_2$:
    \begin{equation}
      \label{eq:cp-res}
    \vxym{
      y_1\ar@{.>}[r]^{g/f}&z\\
      x\ar[u]^-f\ar[r]_-g\ar@{}[ur]|-{\displaystyle\overset\alpha\Longleftrightarrow}&y_2\ar@{.>}[u]_{f/g}
    }
    \end{equation}
  \item there is no infinite path in~$\tilde P_1^\otimes$.
  \end{enumerate}
\end{asm}

\begin{exa}
In Example~\ref{ex:pres-ds2}, the two critical
  pairs between an equational rewriting rule and another rule correspond to the
  relations~$\gen\gamma$ and~$\gen\delta$, and we have
  \begin{align}
    \label{eq:ds2-res}
    \gen{bm}/\gen{ga}&=\gen{mb}\circ\gen{ag}&\gen{ga}/\gen{bm}&=\gen g&
    \gen{na}/\gen{bg}&=\gen{an}\circ\gen{gb}&\gen{bg}/\gen{na}&=\gen g
  \end{align}
  Given a word~$x$ in $P_0^*$, its \emph{transposition number} is the sum, over
  each occurrence of $\gen a$ in~$x$, of the number of occurrences of~$\gen b$ before
  that~$\gen a$. For instance the transposition number of $\gen{babbaa}$ is
  $1+3+3=7$. Given any morphism of the form
  $x\gen{g}y:x\gen{ba}y\to x\gen{ab}y$, the transposition number of $x\gen{ba}y$
  is strictly greater than the one of $x\gen{ab}y$, which shows that there is no
  infinite rewriting path in~$\tilde P_1^\otimes$. Hence
  Assumption~\ref{apt:convergence} is verified.
\end{exa}

\noindent
By the previous discussion, the existence of residuals on critical pairs implies
the existence of residuals of any pair of coinitial rewriting steps.

\begin{lem}
  \label{lem:mon-local-confl}
  Any pair of coinitial rewriting steps,  one of them being equational,
  admits a residual, given as follows:
  \begin{itemize}
  \item given $f:x\to y_1$ and $g:x\to y_2$ forming a critical pair, one of them
    being equational, their residuals are given by
    Assumption~\ref{apt:convergence},
  \item given $f:x\to x'$ and $g:y\to y'$, we have the
     residual
    \begin{align*}
      (xg)/(fy)\qeq fy'
      \qquad\qquad\qquad
      (gx)/(yf)\qeq y'f
    \end{align*}
    with the corresponding relation as in~\eqref{eq:res-exchange},
  \item given $f:y\to y_1$, $g:y\to y_2$ and an object $x$ and $z$, we have
   the residual
    \begin{align*}
      (xgz)/(xfz)\qeq x(g/f)z
    \end{align*}
    with the corresponding relation as in~\eqref{eq:res-context}.
  \end{itemize}
\end{lem}

\noindent
Finally, we extend residuation to any pair of  coinitial rewriting paths,
by Definition~\ref{def:residual}.

\begin{exa}
  \label{ex:mon-res}
  In our  Example~\ref{ex:pres-ds2}, consider the morphism
  \[
  f
  \qeq
  \gen{bm}\circ\gen{naa}
  \qcolon
  \gen{bbaa}
  \qto
  \gen{ba}
  \]
 Its residuals with $\gen{bga}:\gen{bbaa}\to\gen{baba}$ are
  \[
  f/\gen{bga}
  \qeq
  \gen{mb}\circ\gen{ag}\circ\gen{ana}\circ\gen{gba}
  \qquad\qquad\qquad
  \gen{bga}/f\qeq\gen{g}
  \]
  the first one being computed by
  \begin{align*}
    (\gen{bm}\circ\gen{naa})/\gen{bga}&
    =(\gen{bm}/(\gen{bga}/\gen{naa}))\circ(\gen{naa}/\gen{bga})
    =(\gen{bm}/(\gen{bg}/\gen{na})\gen{a})\circ(\gen{na}/\gen{bg})\gen{a}\\&
    =(\gen{bm}/\gen{ga})\circ(\gen{an}\circ\gen{gb})\gen{a}
    =\gen{mb}\circ\gen{ag}\circ\gen{ana}\circ\gen{gba}
  \end{align*}
  using the residuation rules of Definition~\ref{def:residual} and
  relations~\eqref{eq:ds2-res} (this is also illustrated in the third cylinder
  of Example~\ref{ex:pres-ds2-cyl}). In string diagrammatic form, we have
  \[
  f=\fig{res_f}
  \qquad\qquad\qquad\qquad
  f/\gen{bga}=\fig{res_f_res}
  \]
  Note that residuation is defined on rewriting paths (in~$P_1^\otimes$), but we
  did not claim it was well-defined on morphisms in~$P_1^*$. In fact, it is not
  generally compatible with exchange as we now illustrate. Obviously, the
  morphism~$f$ above is equivalent, up to exchange, to the
  morphism~$f'=\gen{na}\circ\gen{bbm}$. But the residuals $f/\gen{bga}$ and
  $f'/\gen{bga}$ are not:
  \[
  f'/\gen{bga}
  \qeq
  \gen{an}\circ\gen{gb}\circ\gen{bmb}\circ\gen{bbm}
  \]
  (see again Example~\ref{ex:pres-ds2-cyl} for details). Graphically,
  \[
  f'=\fig{res_f2}
  \qquad\qquad\qquad
  f'/\gen{bga}=\fig{res_f2_res}
  \]
\end{exa}

\noindent
We recall that, in order for the definition of residual to make sense (\ie for
Lemma~\ref{lem:residual-well-defined} and
Proposition~\ref{prop:residual-relation} to hold), we need a termination
assumption, which directly translates as follows in the monoidal setting:

\begin{asm}
  There is a weight function $\omega_1:P_0^*P_1P_0^*\to N$, where $N$ is a
  noetherian monoid, such that for every rewriting step
  $f\in\tilde P_0^*P_1P_0^*$ and $g\in P_0^*\tilde P_1P_0^*$, we have
  $\omega_1(g/f)<\omega_1(g)$, where we extend the weight as a function
  $\omega_1:P_1^\otimes\to N$ on rewriting paths by
  $\omega_1(g\circ f)=\omega_1(g)+\omega_1(f)$ and $\omega_1(\id)=0$.
\end{asm}

\begin{exa}
  \label{ex:mon-1-termination}
 For our  example, we define a weight
  function
  \[
  \omega_1
  \qcolon
  P_1^\otimes
  \qto
  \N\times\N
  \]
  with $\N\times\N$ equipped with the pointwise sum and lexicographic ordering.
  The weight is defined on rewriting steps by
  \begin{itemize}
  \item $\omega_1(x\gen{m}y)=(p,0)$ where $p$ is the number of occurrences of~$\gen b$
    in~$x$,
  \item $\omega_1(x\gen{n}y)=(p,0)$ where $p$ is the number of occurrences of~$\gen a$
    in~$y$,
  \item $\omega_1(x\gen{g}y)=(0,q)$ where $q$ is the transposition number of~$xy$.
  \end{itemize}
  It is easily checked that the residuals in~\eqref{eq:ds2-res} are strictly
  decreasing:
  \begin{align*}
    (1,0)=\omega_1(\gen{bm})&>\omega_1(\gen{bm}/\gen{ga})=\omega_1(\gen{mb}\circ\gen{ag})=(0,0)\\
    (1,0)=\omega_1(\gen{na})&>\omega_1(\gen{na}/\gen{bg})=\omega_1(\gen{an}\circ\gen{gb})=(0,0)    
  \end{align*}
  Moreover, this also holds for the residuals along equational rewriting steps
  which are obtained by exchange cells, \eg in
  \[
  \vxym{
    \gen{ab}x\gen{aa}\ar[r]^{\gen{ab}x\gen{m}}\ar@{}[dr]|{\displaystyle\overset{\chi_{\gen gx,\gen m}}\Longleftrightarrow}&\gen{ab}x\gen{a}\\
    \gen{ba}x\gen{aa}\ar[u]^{\gen{g}x\gen{aa}}\ar[r]_{\gen{ba}x\gen{m}}&\gen{ba}x\gen{a}\ar[u]_{\gen{g}x\gen{a}}
  }
  \]
  we have $\omega_1(\gen{ba}x\gen{m})>\omega_1(\gen{ab}x\gen{m})$ because the
  first component is the same but the transposition number decreases. Also, the
  order is compatible with left and right actions in the sense that
  $\omega_1(f)>\omega_1(g)$ implies $\omega_1(xfy)>\omega_1(xgy)$. Thus the weight
  $\omega_1$  fulfills our assumption.
\end{exa}


\subsection{The cylinder property}
\label{sec:mon-cyl}
In the previous section, we have explained how the monoidal structure could help
us to handle more easily the existence of residuals: one only has to ensure that
they exist for critical pairs in order to have their existence for pairs of
coinitial rewriting steps. The situation is very similar for the cylinder
property. For instance, suppose that we have a cylinder as on the left.
\begin{equation}
  \label{eq:cylinder-context}
  \vxym{
    x'\ar@/^/@{.>}[rr]^{g_1/f}\ar@/_/@{.>}[rr]_{g_2/f}\ar@{}[rr]|-{\alpha/f}&&y'\\
    \ar[u]^fx\ar@/^/[rr]^{g_1}\ar@/_/[rr]_{g_2}\ar@{}[rr]|-\alpha&&\ar@{.>}[u]_{f/g_1=f/g_2}y
  }
  \qquad\qquad\qquad
  \vxym{
    zx'z'\ar@/^/@{.>}[rr]^{z(g_1/f)z'}\ar@/_/@{.>}[rr]_{z(g_2/f)z'}\ar@{}[rr]|-{z(\alpha/f)z'}&&zy''z'\\
    \ar[u]^{zfz'}zxz'\ar@/^/[rr]^{zg_1z'}\ar@/_/[rr]_{zg_2z'}\ar@{}[rr]|-{z\alpha z'}&&\ar@{.>}[u]_{z(f/g_1)z'=z(f/g_2)z'}zyz'
  }
\end{equation}
Then for every 0-cells $z$ and $z'$, we also have a cylinder as on the right,
which shows that we only have to show the cylinder property for those which are
minimal \wrt contexts on the left and on the right.

Similarly, consider a situation as above where the bottom cell~$\alpha$ is an
exchange rule
\[
\chi_{g_1,g_2}
\qcolon
(y_1g_2)\circ(g_1x_2)
\qTo
(g_1y_2)\circ(x_1g_2)
\qcolon
x_1x_2
\qto
y_1y_2
\]
with $g_1:x_1\to y_1$ and $g_2:x_2\to y_2$. Also, suppose that the vertical
arrow on the left is of the form $fx_2:x_1x_2\to x'_1x_2$ with $f:x_1\to
x_1'$. In this case, one can always complete the cylinder on the left
\begin{equation}
  \label{eq:cyl-triv}
  \vxym{
    x_1'x_2\ar@/^/@{.>}[rr]^{}\ar@/_/@{.>}[rr]_{}\ar@{}[rr]|-{}&&y_1'y_2\\
    \ar[u]^{fx_2}x_1x_2\ar@/^/[rr]^{(y_1g_2)\circ(g_1x_2)}\ar@/_/[rr]_{(g_1y_2)\circ(x_1g_2)}\ar@{}[rr]|-{\chi_{g_1,g_2}}&&\ar@{.>}[u]_{(f/g_1)y_2}y_1y_2
  }
  \qquad\qquad
  \vxym{
    x_1x_2'\ar@/^/@{.>}[rr]^{}\ar@/_/@{.>}[rr]_{}\ar@{}[rr]|-{}&&y_1y_2'\\
    \ar[u]^{x_1f}x_1x_2\ar@/^/[rr]^{(y_1g_2)\circ(g_1x_2)}\ar@/_/[rr]_{(g_1y_2)\circ(x_1g_2)}\ar@{}[rr]|-{\chi_{g_1,g_2}}&&\ar@{.>}[u]_{y_1(f/g_2)}y_1y_2
  }
\end{equation}
as follows (the same argument will of course apply to a cylinder as on the
right). We write
\[
\alpha
\qcolon
(g_1/f)\circ f
\qTo
(f/g_1)\circ g_1
\qcolon
x_1'
\qto
y_1'
\]
for the 2-cell mediating $f$ and $g_1$ with their residual, obtained by
Assumption~\ref{apt:convergence}. The missing cells above are as follows:
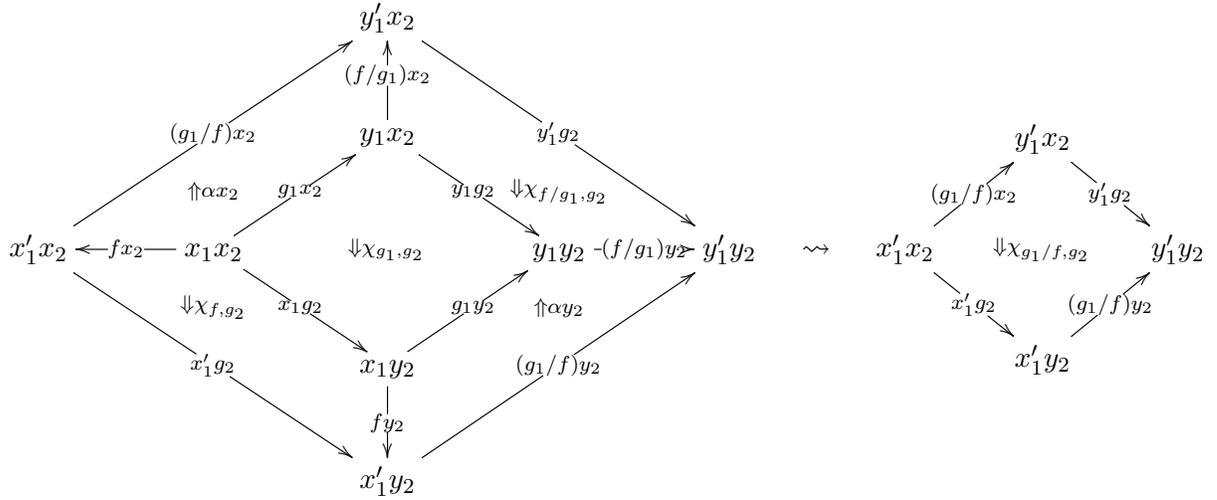
\begin{equation*}
  \vcenter{
    \xymatrix@R=6ex@C=8ex{
      &&y_1'x_2\ar[ddrr]|{y_1'g_2}\\
      &&y_1x_2\ar[dr]|{y_1g_2}\ar[u]|{(f/g_1)x_2}&\\
      x_1'x_2\ar[uurr]|{(g_1/f)x_2}\ar[ddrr]|{x_1'g_2}&\ar[l]|{fx_2}\ar@{}[u]|{\Uparrow\alpha x_2}\ar@{}[d]|{\Downarrow\chi_{f,g_2}}x_1x_2\ar[ur]|{g_1x_2}\ar[dr]|{x_1g_2}\ar@{}[rr]|{\Downarrow\chi_{g_1,g_2}}&&\ar@{}[u]|{\Downarrow\chi_{f/g_1,g_2}}\ar@{}[d]|{\Uparrow\alpha y_2}y_1y_2\ar[r]|-{(f/g_1)y_2}&y_1'y_2\\
      &&x_1y_2\ar[d]|{fy_2}\ar[ur]|{g_1y_2}&\\
      &&x_1'y_2\ar[uurr]|{(g_1/f)y_2}
    }
  }
  \qrsa
  \vcenter{
    \xymatrix{
      &y_1'x_2\ar[dr]|{y_1' g_2}\\
      x_1'x_2\ar[ur]|{(g_1/f)x_2}\ar[dr]|{x_1'g_2}\ar@{}[rr]|{\Downarrow\chi_{g_1/f,g_2}}&&y_1'y_2\\
      &x_1'y_2\ar[ur]|{(g_1/f)y_2}
    }
  }
\end{equation*}

\begin{rem}
  The above diagram should be read as follows, in reference to the notations of
  the cylinder diagram on the left of~\eqref{eq:cylinder-context} (we detail
  this here since this convention will be used again in the following). In the
  center of each picture on the left is figured the 2-cell $\alpha$ (which is
  here~$\chi_{g_1/f,g_2}$), and the morphisms~$f$ and $f/g_1=f/g_2$ are
  represented horizontally as pointing to the left and the right, respectively.
  The rest of the picture on the left exhibits~$g_1/f$ and $g_2/f$.  The
  residual~$\alpha/f$ is displayed on the matching picture on the right:
  \[
  \vxym{
    x'\ar@/^8ex/[rrrr]^{g_1/f}\ar@/_8ex/[rrrr]_{g_2/f}&\ar[l]_-fx\ar@/^3ex/[rr]^{g_1}\ar@/_3ex/[rr]_{g_2}\ar@{}[rr]|{\phantom\alpha\Downarrow\alpha}\ar@{{}{ }{}}@/^6ex/[rr]|\Uparrow\ar@{{}{ }{}}@/_6ex/[rr]|\Downarrow&&y\ar[r]^{f/g_1}_{f/g_2}&y'
  }
  \qrsa
  \vxym{
    x'\ar@/^8ex/[rrrr]^{g_1/f}\ar@/_8ex/[rrrr]_{g_2/f}&\ar@{}[rr]|{\phantom{\alpha/f}\Downarrow{\alpha/f}}&&&y'
  }
  \]
  The ``$\rightsquigarrow$'' sign between the two diagrams indicates here that
  the diagram on the right is the ``top'' of the cylinder whose ``bottom'' and
  ``walls'' are shown on the left; it does not indicate an equality between
  cells, since the diagram on the left cannot be composed and thus does not even
  denote a 2-cell.
\end{rem}


\noindent
The above discussion motivates the introduction of the following definition and
adaptation of the cylinder property.

\begin{defi}
  \label{def:critical-cylinder}
  Suppose given a morphism
  $f:x\to x'$ and a 2-cell $\alpha:g_1\To g_2:x\to y$ in~$P_0^*P_2P_0^*$
  (consisting of one relation in context), as in the left
  of~\eqref{eq:cylinder-context}. Such a pair is \emph{critical} when
  \begin{itemize}
  \item $f$ is different from both $g_1$ and $g_2$,
  \item it is minimal \wrt  contexts: if there is another such pair $(f',\alpha')$
    and $z,z'\in P_0^*$ such that $f=zf'z'$ and $\alpha=z\alpha'z'$ then $z$ and
    $z'$ are both empty,
  \item it is not of the form \eqref{eq:cyl-triv}.
  \end{itemize}
\end{defi}

\begin{rem}
  The critical pairs, in the sense of the previous definition, can easily be
  computed by an adaptation of the usual critical pair algorithm for string
  rewriting systems.
  This is illustrated in Example~\ref{ex:pres-ds2-cyl}.
\end{rem}

\begin{asm}
  The presentation $(P,\tilde P_1)$ satisfies the \emph{cylinder property}: for
  every rewriting step $f:x\to x'$ in $P_0^*\tilde P_1P_0^*$ (\resp in
  $P_0^*P_1P_0^*$) and 2-cell $\alpha:g_1\ToT g_2:x\to y$ in $P_0^*P_2P_0^*$
  with $g_1$ and $g_2$ in $P_0^*P_1P_0^*$ (\resp $P_0^*\tilde P_1P_0^*$) which
  are critical in the sense of Definition~\ref{def:critical-cylinder}, we have
  $f/g_1=f/g_2$ and there exists a 2-cell $g_1/f\overset*\ToT g_2/f$. We write
  $\alpha/f$ for an arbitrary choice of such a 2-cell.
  \begin{equation}
    \label{eq:apt-cylinder}
    \vxym{
      x'\ar@/^/@{.>}[rr]^{g_1/f}\ar@/_/@{.>}[rr]_{g_2/f}\ar@{}[rr]|-{\alpha/f}&&y'\\
      \ar[u]^fx\ar@/^/[rr]^{g_1}\ar@/_/[rr]_{g_2}\ar@{}[rr]|-\alpha&&\ar@{.>}[u]_{f/g_1=f/g_2}y
    }
  \end{equation}
\end{asm}

\noindent
We have restricted the cylinder property to critical pairs in order to have less
computations to perform,
but the previous discussion shows
that the cylinder property holds even for non-critical pairs when the assumption
is valid.

\begin{exa}
  \label{ex:pres-ds2-cyl}
  The presentation of Example~\ref{ex:pres-ds2} satisfies the cylinder property.
 The critical cylinder diagrams are
  \[%
  \,\vxym{%
    &&\gen{aba}\ar[dr]|{\gen{ag}}\\
    &\gen{aaba}\ar[ur]|{\gen{mba}}&\gen{baa}\ar[u]|{\gen{ga}}\ar[dr]|{\gen{bm}}&\gen{aab}\ar[dr]|{\gen{mb}}\\
    \gen{abaa}\ar[ur]|{\gen{aga}}\ar[ddrr]|{\gen{abm}}&\ar@{}[u]|{\Uparrow\gen{\gamma a}}\ar@{}[d]|{\Uparrow\chi_{\gen g,\gen m}}\ar[l]|{\gen{gaa}}\gen{baaa}\ar[ur]|{\gen{bma}}\ar[dr]|{\gen{bam}}\ar@{}[rr]|{\Downarrow\gen{b\alpha}}&&\ar@{}[u]|{\Uparrow\gen\gamma}\ar@{}[d]|{\Downarrow\gen\gamma}\gen{ba}\ar[r]|{\gen g}&\gen{ab}\\
    &&\gen{baa}\ar[d]|{\gen{ga}}\ar[ur]|{\gen{bm}}&\gen{aab}\ar[ur]|{\gen{mb}}\\
    &&\gen{aba}\ar[ur]|{\gen{ag}}
  }%
  \qrsa\,%
  \vxym{%
    &&\gen{aba}\ar@{}[dd]|{\Downarrow\chi_{\gen m,\gen g}}\ar[dr]|{\gen{ag}}\\
    &\gen{aaba}\ar@{}[dddr]|{\Uparrow\gen{a\gamma}}\ar[ur]|{\gen{mba}}\ar[dr]|{\gen{aag}}&&\gen{aab}\ar@{}[dd]|{\Downarrow\gen{\alpha b}}\ar[dr]|{\gen{mb}}\\
    \gen{abaa}\ar[ur]|{\gen{aga}}\ar[ddrr]|{\gen{abm}}&&\gen{aaab}\ar[ur]|{\gen{mab}}\ar[dr]|{\gen{amb}}&&\gen{ab}\\
    &&&\gen{aab}\ar[ur]|{\gen{mb}}\\
    &&\gen{aba}\ar[ur]|{\gen{ag}}
  }%
  \]
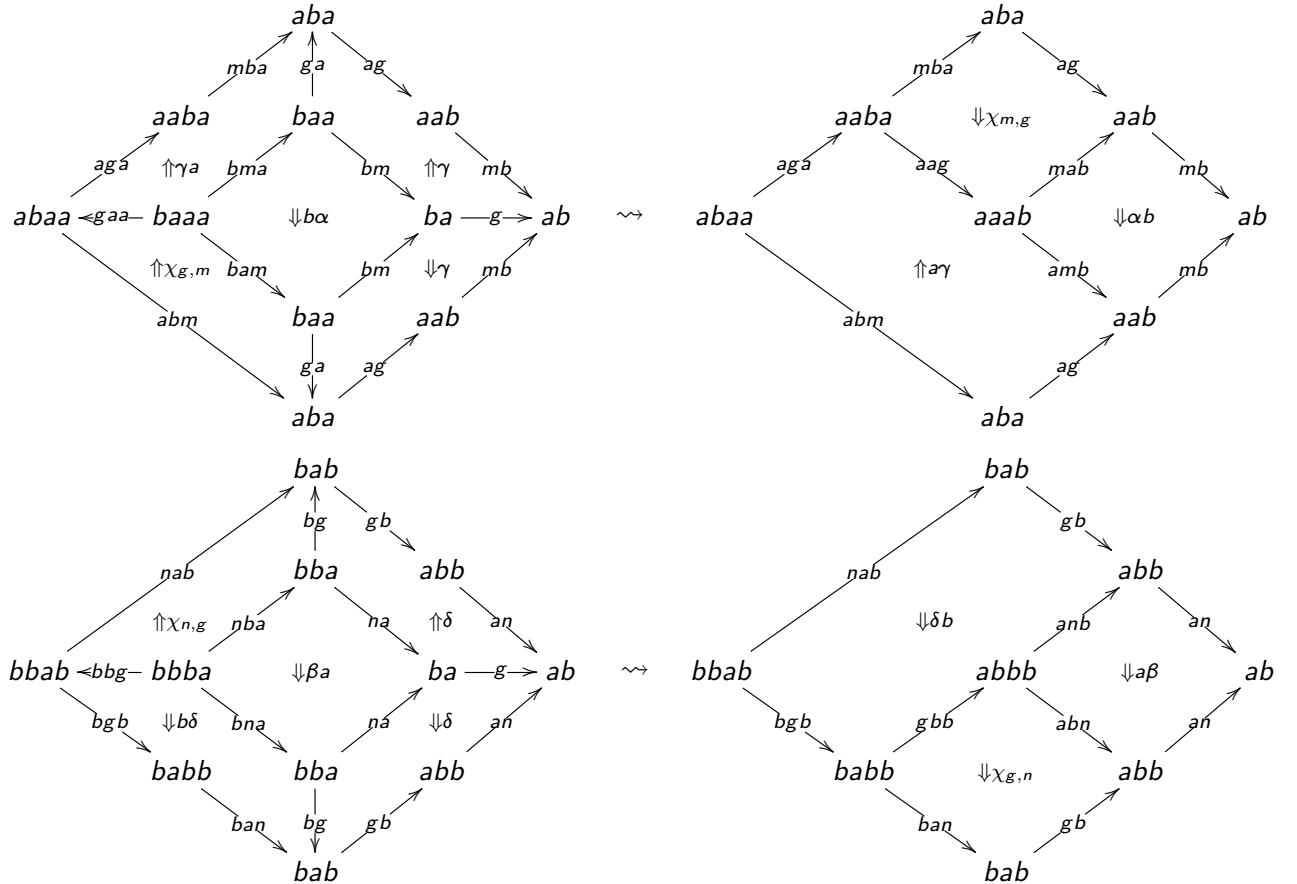
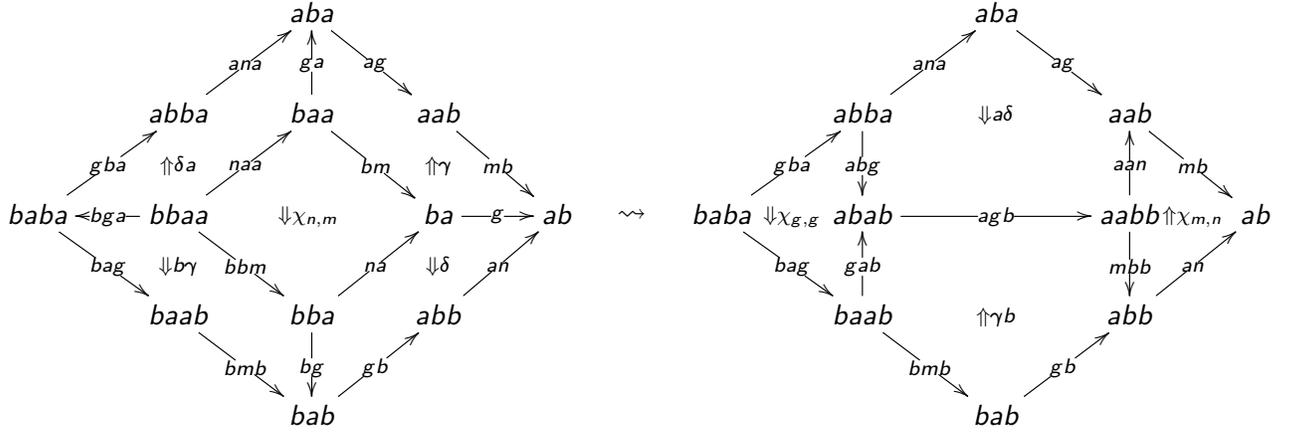
  \[%
  \vxym{%
    &&\gen{bab}\ar[dr]|{\gen{gb}}\\
    &&\gen{bba}\ar[u]|{\gen{bg}}\ar[dr]|{\gen{na}}&\gen{abb}\ar[dr]|{\gen{an}}\\
    \gen{bbab}\ar[uurr]|{\gen{nab}}\ar[dr]|{\gen{bgb}}&\ar[l]|{\gen{bbg}}\ar@{}[u]|{\Uparrow\chi_{\gen n,\gen g}}\ar@{}[d]|{\Downarrow\gen{b\delta}}\gen{bbba}\ar[ur]|{\gen{nba}}\ar[dr]|{\gen{bna}}\ar@{}[rr]|{\Downarrow\gen{\beta a}}&&\ar@{}[u]|{\Uparrow\gen\delta}\ar@{}[d]|{\Downarrow\gen\delta}\gen{ba}\ar[r]|{\gen g}&\gen{ab}\\
    &\gen{babb}\ar[dr]|{\gen{ban}}&\gen{bba}\ar[d]|{\gen{bg}}\ar[ur]|{\gen{na}}&\gen{abb}\ar[ur]|{\gen{an}}\\
    &&\gen{bab}\ar[ur]|{\gen{gb}}\\
  }%
  \qrsa\!%
  \vxym{%
    &&\ar@{}[dddl]|{\Downarrow\gen{\delta b}}\gen{bab}\ar[dr]|{\gen{gb}}\\
    &&&\ar@{}[dd]|{\Downarrow\gen{a\beta}}\gen{abb}\ar[dr]|{\gen{an}}\\
    \gen{bbab}\ar[uurr]|{\gen{nab}}\ar[dr]|{\gen{bgb}}&&\ar@{}[dd]|{\Downarrow\chi_{\gen g,\gen n}}\gen{abbb}\ar[ur]|{\gen{anb}}\ar[dr]|{\gen{abn}}&&\gen{ab}\\
    &\gen{babb}\ar[dr]|{\gen{ban}}\ar[ur]|{\gen{gbb}}&&\gen{abb}\ar[ur]|{\gen{an}}\\
    &&\gen{bab}\ar[ur]|{\gen{gb}}\\
  }\,%
  \]%
  \[%
  \vxym{%
    &&\gen{aba}\ar[dr]|{\gen{ag}}\\
    &\gen{abba}\ar[ur]|{\gen{ana}}&\gen{baa}\ar[u]|{\gen{ga}}\ar[dr]|{\gen{bm}}&\gen{aab}\ar[dr]|{\gen{mb}}\\
    \gen{baba}\ar[ur]|{\gen{gba}}\ar[dr]|{\gen{bag}}&\ar@{}[u]|{\Uparrow\gen{\delta a}}\ar@{}[d]|{\Downarrow\gen{b\gamma}}\ar[l]|{\gen{bga}}\gen{bbaa}\ar[ur]|{\gen{naa}}\ar[dr]|{\gen{bbm}}\ar@{}[rr]|{\Downarrow\chi_{\gen n,\gen m}}&&\ar@{}[u]|{\Uparrow\gen\gamma}\ar@{}[d]|{\Downarrow\gen\delta}\gen{ba}\ar[r]|{\gen g}&\gen{ab}\\
    &\gen{baab}\ar[dr]|{\gen{bmb}}&\gen{bba}\ar[d]|{\gen{bg}}\ar[ur]|{\gen{na}}&\gen{abb}\ar[ur]|{\gen{an}}\\
    &&\gen{bab}\ar[ur]|{\gen{gb}}
  }%
  \qrsa%
  \vxym{%
    &&\gen{aba}\ar[dr]|{\gen{ag}}\ar@{}[dd]|{\Downarrow\gen{a\delta}}\\
    &\gen{abba}\ar[ur]|{\gen{ana}}\ar[d]|{\gen{abg}}&&\gen{aab}\ar[dr]|{\gen{mb}}\\
    \gen{baba}\ar[ur]|{\gen{gba}}\ar[dr]|{\gen{bag}}\ar@{}[r]|{\Downarrow\chi_{\gen g,\gen g}}&\gen{abab}\ar[rr]|{\gen{agb}}&&\gen{aabb}\ar[u]|{\gen{aan}}\ar[d]|{\gen{mbb}}\ar@{}[r]|{\Uparrow\chi_{\gen m,\gen n}}&\gen{ab}\\
    &\gen{baab}\ar[dr]|{\gen{bmb}}\ar[u]|{\gen{gab}}&&\gen{abb}\ar[ur]|{\gen{an}}\\
    &&\gen{bab}\ar[ur]|{\gen{gb}}\ar@{}[uu]|{\Uparrow\gen{\gamma b}}
  }%
  \]
  Let us explain how these were computed. Given a critical cylinder as
  in~\eqref{eq:apt-cylinder}, either the vertical morphism ($f$) or the
  horizontal arrows ($g_1$ and $g_2$ at the source and target of the relation
  $\alpha$) are equational:
  \begin{itemize}
  \item if the vertical arrow is equational, then it is of the form
    $x\gen{g}y:x\gen{ba}y\to x\gen{ab}y$; therefore the horizontal relation
    should have a $0$-source which ``intersects'' $\gen{ba}$ in a non-trivial
    way; this source thus either
    \begin{itemize}
    \item begins by an $\gen a$: this gives rise to the first cylinder,
    \item ends by a $\gen b$: this gives rise to the second cylinder, or
    \item contains $\gen{ba}$: this gives rise to the third cylinder.
    \end{itemize}
  \item if the horizontal arrows are equational, then the horizontal relation is
    necessarily an exchange between two morphisms of the form $x\gen{g}y$,
    because no generating relation has equational source and target. We can then
    examine all the possibilities for $x$ and~$y$, and vertical rewriting steps,
    and show that they are all trivial, for instance
    \[
    \vxym{
      \gen{baba}&&\\
      \ar[u]^{\gen{naba}}\gen{bbaba}\ar@/^/[rr]^{}\ar@/_/[rr]_{}\ar@{}[rr]|-{\gen b\chi_{\gen g,\gen g}}&&\gen{babab}
    }
    \qquad\qquad\qquad
    \vxym{
      \gen{baba}&&\\
      \ar[u]^{\gen{bmba}}\gen{baaba}\ar@/^/[rr]^{}\ar@/_/[rr]_{}\ar@{}[rr]|-{\chi_{\gen{ga},\gen g}}&&\gen{abaab}
    }
    \]
    are both of the form~\eqref{eq:cyl-triv}.
  \end{itemize}
\end{exa}


\noindent
In order for the global cylinder property to hold
(Proposition~\ref{prop:global-cylinder}), we need again a termination
assumption, which can be reformulated as follows in the monoidal setting.

\begin{asm}
  There is a weight function $\omega_2:P_0^*P_2P_0^*\to N$, where $N$ is a
  noetherian commutative monoid, such that for every $\alpha:g_1\To g_2$
  in~$P_2$ and $f$ in $P_1$ such that $\alpha/f$ exists, we have
  $\omega_2(\alpha/f)<\omega_2(\alpha)$, where $\omega_2$ is extended to
  arbitrary 2-cells by acting the same on inverses, sending both compositions to
  addition and identities to the neutral element.
\end{asm}

\begin{exa}
  \label{ex:mon-2-termination}
  Going back to Example~\ref{ex:pres-ds2-cyl}, we define
  $\omega_2:P_2\to\N\times\N$ in a similar way as in
  Example~\ref{ex:mon-1-termination} by
  \begin{itemize}
  \item $\omega_2(x\gen\alpha y)=(p,0)$ where $p$ is the number of occurrences
    of~$\gen b$ in~$x$,
  \item $\omega_2(x\gen\beta y)=(p,0)$ where $p$ is the number of occurrences
    of~$\gen a$ in~$y$,
  \item $\omega_2(x\chi y)=(0,q)$ where $q$ is the transposition number of~$xy$.
  \end{itemize}
  It is easy to check that this interpretation is compatible with contexts, \ie
  $\omega_2(\alpha)>\omega_2(\beta)$ implies
  $\omega_2(x\alpha y)>\omega_2(x\beta y)$, that the cylinders of
  Example~\ref{ex:pres-ds2-cyl} are strictly decreasing (the ``top'' is smaller
  than the ``bottom''), and that residual of exchange relations are decreasing.
\end{exa}

\noindent
The global cylinder property folllows from these assumptions (replacing $P_1^*$
by $\tilde P_1^\otimes$ in Proposition~\ref{prop:global-cylinder}, see
Example~\ref{ex:mon-res}), as well as other properties mentioned in
Section~\ref{sec:cylinder}.
Moreover, Theorem~\ref{thm:nf-quotient} holds in our context, in a way which is
compatible with monoidal structure. Namely, the category of normal forms is
monoidal with the tensor product defined on objects by
\[
\nf{x}\otimes\nf{y}
\qeq
\widehat{\nf{x}\nf{y}}
\]
and the action of objects $\nf{x}$ and $\nf{z}$ on a morphism $f:y\to y'$ is
defined (similarly to the proof of Theorem~\ref{thm:nf-quotient}) by
\[
\nf{x}\otimes f\otimes\nf{z}
\qeq
u_{y''}\circ(\nf{x}f\nf{z}/u_{\nf{x}\nf{y}\nf{z}})
\]
where $y''$ is the target of $\nf{x}f\nf{z}/u_{\nf{x}\nf{y}\nf{z}}$. Graphically,
\[
\xymatrix@C=8ex{
  &\widehat{\nf x\nf y'\nf z}\\
  \widehat{\nf x\nf y\nf z}\ar[r]^{\nf{x}f\nf{z}/u_{\nf{x}\nf{y}\nf{z}}}\ar@/^/@{.>}[ur]^{\nf{x}\otimes f\otimes\nf{y}}\ar@{}[dr]|{\displaystyle\overset*\ToT}&y''\ar[u]_{u_{y''}}\\
  \nf{x}\nf{y}\nf{z}\ar[u]^{u_{\nf{x}\nf{y}\nf{z}}}\ar[r]_{\nf{x}f\nf{z}}&\nf{x}\nf{y}'\nf{z}\ar[u]_{u_{\nf{x}\nf{y}\nf{z}}/\nf{x}f\nf{z}}
}
\]
Similarly, the quotient category is monoidal as a quotient of a monoidal
category by a congruence respecting tensor product.

\begin{thm}
  \label{thm:mon-nf-quot}
  The canonical monoidal functor $\nfcat{P}{\tilde P_1}\to\pcat P/\tilde P_1$ is
  a monoidal isomorphism of categories.
\end{thm}

\subsection{The coherence theorem}
\label{sec:mon-coh}
We can finally extend the coherence Theorem~\ref{thm:quotient-loc}, by verifying
that it is compatible with the monoidal structure of the categories:

\begin{thm}
  \label{thm:mon-quotient-loc}
  A presentation modulo $(P,\tilde P_2)$ which satisfies
  Assumptions~\ref{apt:convergence} to \ref{apt:2-termination} is coherent, in
  the sense that there exists a pair of functors
  \[
  F
  \qcolon
  \pcat P/\tilde P_1
  \quad\rightleftarrows\quad
  \loc{\pcat P}{\tilde P_1}
  \qcolon
  G
  \]
  forming an equivalence of categories, with $F$ strong monoidal and $G$ strict
  monoidal.
\end{thm}
\begin{proof}
  The functors are constructed in the proofs of Theorems~\ref{thm:nf-quotient}
  and~\ref{thm:quotient-loc}; we only have to check that they are monoidal. We
  have $F(\monunit)=\monunit$, so we can take $\eta=\id_I$. Given two
  objects~$\nf{x}$ and~$\nf{y}$ in $\nfcat{P}{\tilde P_1}$ (which is monoidally
  isomorphic to $\pcat P/\tilde P_1$ by Theorem~\ref{thm:mon-nf-quot}), we have
  $F(\nf{x})\otimes F(\nf{y})=\nf{x}\nf{y}$ and
  $F(\nf{x}\otimes\nf{y})=\widehat{\nf{x}\nf{y}}$. There exists a normalization
  path $u:\nf{x}\nf{y}\to\widehat{\nf{x}\nf{y}}$ in $\pcat{P}$ and we define
  $\mu_{x,y}=Lu$, where $L:\pcat{P}\to\loc{\pcat P}{\tilde P_1}$ is the
  localization functor, and $\mu_{x,y}$ is invertible because~$u$ is
  equational. The axioms for monoidal functors are easily verified by
  convergence of the equational rewriting system
  (Assumption~\ref{apt:convergence}). For instance, the first diagram of
  Definition~\ref{def:mon-funct} boils down to
  \[
  \vxym{
    \widehat{\nf{x}\nf{y}}\nf{z}\ar[r]\ar@{}[dr]|{\displaystyle\overset*\Leftrightarrow}&\widehat{\nf{x}\nf{y}\nf{z}}\\
    \nf{x}\nf{y}\nf{z}\ar[u]\ar[r]&\nf{x}\widehat{\nf{y}\nf{z}}\ar[u]
  }
  \]
  which follows from Newman's Lemma~\ref{lem:newman}. Conversely, the functor
  $G$ is defined on objects by~$G(x)=\nf{x}$ so that we have
  $G(\monunit)=\monunit$ and
  \[
  G(x)\otimes G(y)
  \qeq
  \widehat{\nf{x}\nf{y}}
  \qeq
  \widehat{xy}
  \qeq
  G(x\otimes y)
  \]
  from which we deduce that we can take $\eta=\id_\monunit$ and
  $\mu_{x,y}=\id_{\widehat{xy}}$.
\end{proof}

In particular, the presentation of Example~\ref{ex:pres-ds2} is coherent.

\subsection{A variant of the cylinder property}
\label{sec:mon-cyl'}
As we saw in Example~\ref{ex:mon-res}, residuation is not in general compatible
with exchange, so that we cannot expect the cylinder property (and
Assumption~\ref{apt:cylinder} in particular) to hold in every case. In fact, a
reasonable generalization of the global cylinder property
(Proposition~\ref{prop:global-cylinder}) could be: given coinitial
morphisms~$f:x\to x'$ in~$\tilde P_1^*$ (\resp in $P_1^*$) and $g_1,g_2:x\to y$
in~$P_1^*$ (\resp in $\tilde P_1^*$) such that there exists a composite 2-cell
$\alpha:g_1\overset*\ToT g_2$, we have $f/g_1\overset*\ToT f/g_2$ and there
exists a 2-cell $g_1/f\overset*\ToT g_2/f$.
\[
\vxym{
  x'\ar@/^/@{.>}[rr]^{g_1/f}\ar@/_/@{.>}[rr]_{g_2/f}\ar@{}[rr]|-{\phantom *\Updownarrow*}&&y'\\
  \ar[u]^fx\ar@/^/[rr]^{g_1}\ar@/_/[rr]_{g_2}\ar@{}[rr]|-{\phantom *\Updownarrow*}&&\ar@{.>}@/^/[u]^{f/g_1}\ar@{}[u]|{\overset*\ToT}\ar@{.>}@/_/[u]_{f/g_2}y
}
\]
Note that we do not require $g_1/f$ and $g_2/f$ to be equal, but only merely
equivalent. However, such a general global cylinder property seems to be
difficult to be deduced from a local property that would  generalize
Assumption~\ref{apt:cylinder} and could easily be checked in practice, so that we
have to restrict to particular cases for now. As an illustration, in this
section, we study the dual of the presentation modulo of
Example~\ref{ex:pres-ds2} and show that it can be handled using a a different 
local cylinder property.

\begin{exa}
  \label{ex:pres-ds2-op}
  We write now~$P$ for the opposite of the presentation of
  Example~\ref{ex:pres-ds2}:
  it has $P_0=\set{\gen a,\gen b}$ as set of generators for objects and the
  generators for morphisms are the dual of those of Example~\ref{ex:pres-ds2}
  (we write $\ol f$ for the dual of a generator~$f$):
  \[
  P_1\qeq\set{\ol{\gen m}:\gen a\to\gen{aa},\ol{\gen n}:\gen b\to\gen{bb},\ol{\gen g}:\gen{ab}\to\gen{ba}}
  \]
  where $\ol{\gen g}$ is the only equational generator:
  $\tilde P_1=\set{\ol{\gen g}}$.  We would like to show that this presentation
  satisfies Assumptions~\ref{apt:convergence} to~\ref{apt:2-termination}, in
  order to be able to apply our main Theorem~\ref{thm:mon-quotient-loc}. Notice
  that, here, it is important that the termination assumptions are restricted to
  equational morphisms, since there is no hope to have a terminating rewriting
  system with all generators. For instance, we have
  \[
  \gen a
  \overset{\ol{\gen m}}\longrightarrow
  \gen{aa}
  \overset{\ol{\gen m}\gen a}\longrightarrow
  \gen{aaa}
  \overset{\ol{\gen m}\gen{aa}}\longrightarrow
  \gen{aaaa}
  \overset{\ldots}\longrightarrow
  \ldots
  \]
  The relations in~$P_2$ are the dual of those of Example~\ref{ex:pres-ds2}:
  \begin{align*}
    \ol{\gen\alpha}\qcolon \ol{\gen m}\gen a\circ\ol{\gen m}&\qTo\gen a\ol{\gen m}\circ\ol{\gen m}\\
    \ol{\gen\beta}\qcolon \ol{\gen n}\gen b\circ\ol{\gen n}&\qTo\gen b\ol{\gen n}\circ\ol{\gen n}\\
    \ol{\gen\gamma}\qcolon\gen b\ol{\gen m}\circ\ol{\gen g}&\qTo\ol{\gen g}\gen a\circ\gen a\ol{\gen g}\circ\ol{\gen m}\gen b\\
    \ol{\gen\delta}\qcolon\ol{\gen n}\gen a\circ\ol{\gen g}&\qTo\gen b\ol{\gen g}\circ\ol{\gen g}\gen b\circ\gen a\ol{\gen n}
  \end{align*}
  The orientation of the source (\resp target) cell has been reversed, and the
  orientation of the relation does not really matter here since we are
  interested in the generated equivalence relation (here, we chose to keep the
  same orientation). Assumption~\ref{apt:convergence} can be checked by
  constructing the two critical residuation squares:
  \[
  \vxym{
    \gen{aab}\ar[r]^{\gen a\ol{\gen g}}\ar@{}[drr]|{\displaystyle\overset{\ol{\gen\gamma}}\Leftrightarrow}&\gen{aba}\ar[r]^{\ol{\gen g}\gen a}&\gen{baa}\\
    \gen{ab}\ar[u]^{\ol{\gen m}\gen b}\ar[rr]_{\ol{\gen g}}&&\gen{ba}\ar[u]_{\gen b\ol{\gen m}}
  }
  \qquad\qquad\qquad
  \vxym{
    \gen{abb}\ar[r]^{\ol{\gen g}\gen b}\ar@{}[drr]|{\displaystyle\overset{\ol{\gen\delta}}\Leftrightarrow}&\gen{bab}\ar[r]^{\gen b\ol{\gen g}}&\gen{bba}\\
    \gen{ab}\ar[u]^{\gen a\ol{\gen n}}\ar[rr]_{\ol{\gen g}}&&\gen{ba}\ar[u]_{\ol{\gen n}\gen a}
  }
  \]
  Termination of the equational rewriting system is easily checked using a
  transposition number as before (counting now the number of occurrences of
  $\gen b$ after occurrences of~$\gen a$). Assumption~\ref{apt:1-termination}
  can be checked by a variation of Example~\ref{ex:mon-1-termination} (obtained
  by exchanging the role of $\gen a$ and $\gen b$ in $\omega_1$). However, there
  is no hope that Assumption~\ref{apt:cylinder} will hold. Namely, consider the
  ``cylinder'' formed by $\ol{\gen g}$ and $\chi_{\ol{\gen m},\ol{\gen n}}$ as
  depicted below:
  \begin{multline}
    \label{eq:pres-ds2-op-cyl}
    \vxym{
      &&\gen{baa}\ar@/^7ex/[rrrrrddd]|{\ol{\gen n}\gen{aa}}\ar@{}[ddrrr]|{\Uparrow\ol{\gen\delta}\gen a}&&&\\
      &&\gen{aba}\ar[u]|{\ol{\gen g}\gen a}\ar@/^/[drrr]|{\gen a\ol{\gen n}\gen a}\\
      &&\gen{aab}\ar@{}[dd]|{\Downarrow\chi_{\ol{\gen m},\ol{\gen n}}}\ar[u]|{\gen a\ol{\gen g}}\ar[dr]|{\gen{aa}\ol{\gen n}}\ar@{}[rrr]|{\Uparrow\gen a\ol{\gen\delta}}&&&\gen{abba}\ar[dr]|{\ol{\gen g}\gen{ba}}\\
      \gen{ba}\ar@/^/[uuurr]|{\gen{b}\ol{\gen m}}\ar@/_/[dddrr]|{\ol{\gen n}\gen a}&\ar[l]|{\ol{\gen g}}\ar@{}[uu]|{\Uparrow\ol{\gen\gamma}}\ar@{}[dd]|{\Downarrow\ol{\gen\delta}}\gen{ab}\ar[ur]|{\ol{\gen m}\gen b}\ar[dr]|{\gen a\ol{\gen n}}&&\gen{aabb}\ar[r]|<<<<{\gen a\ol{\gen g}\gen b}&\gen{abab}\ar[ur]|{\gen{ab}\ol{\gen g}}\ar[dr]|{\ol{\gen g}\gen{ab}}&&\gen{baba}\ar[r]|<<<<{\gen b\ol{\gen g}\gen a}&\gen{bbaa}\\
      &&\gen{abb}\ar[d]|{\ol{\gen g}\gen b}\ar[ur]|{\ol{\gen m}\gen{bb}}\ar@{}[ddrrr]|{\Downarrow\gen b\ol{\gen\gamma}}\ar@{}[rrr]|{\Downarrow\ol{\gen\gamma}\gen b}&&&\gen{baab}\ar@{}[uu]|{\Uparrow\chi_{\ol{\gen g},\ol{\gen g}}}\ar[ur]|{\gen{ba}\ol{\gen g}}\\
      &&\gen{bab}\ar[d]|{\gen b\ol{\gen g}}\ar@/_/[urrr]|{\gen b\ol{\gen m}\gen b}\\
      &&\gen{bba}\ar@/_7ex/[rrrrruuu]|{\gen{bb}\ol{\gen m}}&&&
    }
    \displaybreak[0]\\
    \qrsa
    \vxym{
      &\gen{baa}\ar[dr]|{\ol{\gen n}\gen{aa}}&\\
      \gen{ba}\ar[ur]|{\gen b\ol{\gen m}}\ar[dr]|{\ol{\gen n}\gen a}\ar@{}[rr]|{\Uparrow\chi_{\ol{\gen n},\ol{\gen m}}}&&\gen{bbaa}\\
      &\gen{bba}\ar[ur]|{\gen{bb}\ol{\gen m}}
    }
  \end{multline}
  This is not a proper cylinder because we have
  \[
  \ol{\gen g}/(\gen{aa}\ol{\gen n}\circ \ol{\gen m}\gen b)
  \qeq
  \gen b\ol{\gen g}\gen a\circ\ol{\gen g}\gen{ba}\circ\gen{ab}\ol{\gen g}\circ\gen a\ol{\gen g}\gen b
  \qneq
  \gen b\ol{\gen g}\gen a\circ\gen{ba}\ol{\gen g}\circ\ol{\gen g}\gen{ab}\circ\gen a\ol{\gen g}\gen b
  \qeq
  \ol{\gen g}/(\ol{\gen m}\gen{bb}\circ\gen a\ol{\gen n})
  \]
  The two morphisms in the middle are not equal, they are only equivalent up to
  exchange (up to the relation~$\exch$). Also notice that the residual of the
  exchange relation $\chi_{\ol{\gen m},\ol{\gen n}}$ after $\ol{\gen g}$ is an
  exchange relation ($\chi_{\ol{\gen n},\ol{\gen m}}$ as pictured on the right
  in the above figure).
\end{exa}

\noindent
This example suggests modifying Assumption~\ref{apt:cylinder} to

\renewcommand\theassumption{\ref{apt:cylinder}'}
\begin{asm}
  \label{apt:cylinder-exch}
  The presentation $(P,\tilde P_1)$ satisfies the following conditions.
  \begin{enumerate}
  \item The cylinder property holds up to $\exch$: for every rewriting step
    $f:x\to x'$ in $P_0^*\tilde P_1P_0^*$ (\resp in $P_0^*P_1P_0^*$) and 2-cell
    $\alpha:g_1\ToT g_2:x\to y$ in $P_0^*P_2P_0^*$ with $g_1$ and $g_2$ in
    $P_0^*P_1P_0^*$ (\resp $P_0^*\tilde P_1P_0^*$) which are critical in the
    sense of Definition~\ref{def:critical-cylinder}, we have $f/g_1\exch f/g_2$
    and there exists a 2-cell $g_1/f\overset*\ToT g_2/f$. We write $\alpha/f$
    for an arbitrary choice of such a 2-cell.
    \begin{equation*}
      \vxym{
        x'\ar@/^/@{.>}[rr]^{g_1/f}\ar@/_/@{.>}[rr]_{g_2/f}\ar@{}[rr]|-{\alpha/f}&&y'\\
        \ar[u]^fx\ar@/^/[rr]^{g_1}\ar@/_/[rr]_{g_2}\ar@{}[rr]|-\alpha&&\ar@{.>}@/^/[u]^{f/g_1}\ar@{}[u]|{\exch}\ar@{.>}@/_/[u]_{f/g_2}y
      }
    \end{equation*}
  \item Residuation is compatible with the relation $\exch$: in the cases above
    where $\alpha$ is an exchange cell in context, its residual~$\alpha/f$ is
    also a composite of exchange cells in context.
  \end{enumerate}
\end{asm}

\noindent
Note that the second condition implies that we can consider morphisms up to
exchange, and compute their residuals:

\begin{lem}
  For every coinitial morphisms $f,f'$ and $g,g'$ such that $f\exch f'$ and
  $g\exch g'$, with $f$ and $f'$ equational, we have $g/f\exch g'/f'$.
\end{lem}

\begin{rem}
  In the presentation of Example~\ref{ex:pres-ds2}, residuation is not
  compatible with exchange because the last cylinder of
  Example~\ref{ex:pres-ds2-cyl} shows that an exchange relation can have a
  residual which does not consist of exchange relations (in context) only. 
Thus, while 
  first condition of Assumption~\ref{apt:cylinder-exch} is a relaxed version of
  Assumption~\ref{apt:cylinder},  the second condition is a
  strengthening, and the two assumptions are thus incomparable.
\end{rem}

\begin{exa}
  In our Example~\ref{ex:pres-ds2-op}, one easily checks that the only critical
  cylinder is \eqref{eq:pres-ds2-op-cyl}. The only residuals of exchange
  relations are thus of the form~\eqref{eq:cyl-triv} (in context) and are
  therefore exchange relations in context: residuation is compatible with
  $\exch$. As mentioned before, the critical cylinder~\eqref{eq:pres-ds2-op-cyl}
  is of the right shape, up to exchange. For the termination
  Assumption~\ref{apt:2-termination}, we distinguish two cases depending on
  whether the vertical arrow is equational or not, as explained in
  Remark~\ref{rem:2-termination}. When the vertical arrow is equational (\ie of
  the form $x\ol{\gen g}y$), termination is shown using the variant of
  Example~\ref{ex:mon-2-termination} obtained by exchanging the role of $\gen a$
  and $\gen b$ in $\omega_2$. However, this same weight~$\omega_2$ will not work
  when the vertical arrow is not equational.
  For instance, the residual of the relation
  \[
  \gen{ab}\chi_{\ol{\gen g},\ol{\gen g}}
  \qcolon
  \gen{abba}\ol{\gen g}\circ\gen{ab}\ol{\gen g}\gen{ab}
  \qTo
  \gen{ab}\ol{\gen g}\gen{ba}\circ\gen{abab}\ol{\gen g}
  \qcolon
  \gen{ababab}
  \qto
  \gen{abbaba}
  \]
  after the morphism
  \[
  \ol{\gen m}\gen{babab}
  \qcolon
  \gen{ababab}
  \qto
  \gen{aababab}
  \]
  is
  \[
  \gen{aab}\chi_{\ol{\gen g},\ol{\gen g}}
  \qcolon
  \gen{aabba}\ol{\gen g}\circ\gen{aab}\ol{\gen g}\gen{ab}
  \qTo
  \gen{aab}\ol{\gen g}\gen{ba}\circ\gen{aabab}\ol{\gen g}
  \qcolon
  \gen{aababab}
  \qto
  \gen{aabbaba}
  \]
  and we have
  \[
  (0,1)
  \qeq
  \omega_1(\gen{ab}\chi_{\ol{\gen g},\ol{\gen g}})
  \quad\not>\quad
  \omega_1(\gen{aab}\chi_{\ol{\gen g},\ol{\gen g}})
  \qeq
  (0,2)
  \]
  Intuitively, in $\gen{ab}\chi_{\ol{\gen g},\ol{\gen g}}$ there is one
  transposition left to do in the context, whereas after residuation the
  $\gen a$ was duplicated and therefore there are two transpositions left in the
  context in $\gen{aab}\chi_{\ol{\gen g},\ol{\gen g}}$.  However, it can be
  noticed that in cases of the form~\eqref{eq:cyl-triv} (in context) the
  residual of the vertical rewriting step is always a rewriting step (as opposed
  to a rewriting path) and therefore the global cylinder property can be shown
  as explained in Remark~\ref{rem:2-termination}. Finally, it can be shown as in
  Theorem~\ref{thm:mon-quotient-loc} that the presentation
  modulo~$(P,\tilde P_1)$ is coherent.
\end{exa}

\section{Conclusion}
\label{sec:concl}
We have introduced a notion of presentation of a (monoidal) category modulo an
``equational'' rewriting system, and provided
coherence conditions ensuring that the equational rules are well-behaved \wrt
the generators. In particular, we show that, under those assumptions, all the
three possible natural constructions for the presented category are
equivalent. These assumptions are ``local'' in the sense that they are given
directly on the presentations, and can thus be used in practice in order to
perform computations, as illustrated in the article. A more general theory of
situations where quotient coincides with localization is left for future work.

In the future, we would like to investigate more applications, by studying
generic situations. For instance, given two monoidal categories with a coherent
presentation, can we always construct a monoidal presentation of their product?
Having more illustrative examples is also important to evaluate how generic the
assumptions we proposed are. As we have explained in Section~\ref{sec:mon-cyl},
the general methodology seems to be quite stable, but there are many possible
local conditions in order to implement it (\eg local cylinder assumptions such
as Assumptions~\ref{apt:cylinder} or \ref{apt:cylinder-exch} in order to show
the global cylinder property). In particular, we would like to have more general
conditions which would encompass both Assumptions~\ref{apt:cylinder}
and~\ref{apt:cylinder-exch}.
On the practical side, it would be interesting to study extensions of the
Knuth-Bendix procedure which could transform a presentation in order to
hopefully complete it into one satisfying our assumptions.
Finally, we would like to study applications to coherence of various algebraic
structures: presentations modulo allow one to turn some of the generators into
isomorphisms, while remaining equivalent to the situation where those generators
are identities, which is what the coherence theorems (such as MacLane's theorem
for monoidal categories) ensure, in a slightly different formal context.


\bibliographystyle{abbrv}
\bibliography{papers}
\end{document}